\documentclass[fleqn,12pt]{article}
\usepackage{natbib}
\usepackage{amsthm,amsfonts,amsopn,amsmath,amssymb}
\usepackage{pgf,pgfplots,tikz}
\usepackage{bbm,bm}

\theoremstyle{definition}
\newtheorem{theorem}{Theorem}

\newtheorem{definition}{Definition}
\newtheorem{example}{Example}

\RequirePackage[pagebackref=true]{hyperref}
\hypersetup{colorlinks=true,linkcolor=blue,urlcolor=blue,citecolor=red,
    pdftitle=Identifying Present-biased Discount Functions in Dynamic Discrete Choice Models,
    pdfauthor=Jaap Abbring and Oeystein Daljord and Fedor Iskhakov,
    pdfsubject=JEL Codes C14 C25 D91 D92,
    pdfkeywords=discount factor dynamic discrete choice identification present bias quasi-hyperbolic  time preferences
    pdfdisplaydoctitle=true}

\usepackage[nohead]{geometry} 
\usepackage[onehalfspacing]{setspace}
\usepackage{chngcntr}
\newcommand{\cites}[1]{\citeauthor{#1}'s (\citeyear{#1})}

\newcommand{\cp}{s}
\newcommand{\CP}{S}

\begin{document}

\title{Identifying Present-Biased Discount Functions in Dynamic Discrete Choice Models\protect\thanks{Thanks to Hanming Fang for helpful discussions.}}

\author{Jaap H. Abbring\thanks{Department of Econometrics \& OR, Tilburg University, P.O. Box
90153, 5000 LE Tilburg, The Netherlands. E-mail: \href{mailto:jaap@abbring.org}{jaap@abbring.org}. Web: \href{http://jaap.abbring.org}{jaap.abbring.org}}\and
\O ystein Daljord\thanks{Booth School of Business, The University of Chicago.}
\and
Fedor Iskhakov\thanks{1018 HW Arndt building
The Australian National University, Canberra, ACT 0200, Australia E-mail:
\href{mailto:fedor.iskhakov@anu.edu.au}{fedor.iskhakov@anu.edu.au}. Web: \href{http://fedor.iskh.me/}{http://fedor.iskh.me/}
}}

\date{May 27, 2020\thanks{\O ystein \href{https://www.chicagobooth.edu/why-booth/stories/remembering-oystein-daljord}{passed away on June 13, 2020}. This is the version of our paper in \O ystein's last (May 27, 2020) commit to our project's Github repository. We only edited it to include these comments and to suppress Section \ref{s:margarine}'s empirical results, which were computed by \O ystein on a processed version of the Kilts-Nielsen data that we currently (July 2025) cannot access and that have not been cleared for public dissemination. This version was presented at various seminars and conferences, including the 2020 World Congress of the Econometric Society. An earlier (August 2019) draft appeared on \O ystein's faculty page at Booth. 
}}
\maketitle
\abstract{
We study the identification of dynamic discrete choice models with sophisticated, quasi-hyperbolic time preferences  under exclusion restrictions. We consider both standard finite horizon problems and empirically useful infinite horizon ones, which we prove to always have solutions. We reduce identification to finding the present-bias and standard discount factors that solve a system of polynomial equations with coefficients determined by the data and use this to bound the cardinality of the identified set. The discount factors are usually identified, but hard to precisely estimate, because exclusion restrictions do not capture the defining feature of present bias, preference reversals, well.
}

\thispagestyle{empty}
\clearpage

\section{Introduction}

The standard experimental approach to measure dynamically-consistent time preferences identifies them from choice responses to variation in future values, holding current payoffs fixed \citep[see][for surveys]{jel02:fredericketal,urminskyzauberman15}. Observational studies using dynamic discrete choice models with dynamically-consistent preferences have invoked a similar intuition \citep[e.g.][]{yaoetal12,Lee2013,chungetal14, bollinger15,aer19:degrooteverboven}. \cite{qe20:abbringdaljord} formalized this intuition by showing that, in such models, exclusion restrictions on utility set identify the geometric discount factor and, nonparametrically, the utility function. We extend \citeauthor{qe20:abbringdaljord}'s approach to dynamic discrete choice models with sophisticated quasi-hyperbolic discounting, which is represented by a {\em present-bias factor} $\beta$ and a {\em standard discount factor} $\delta$. \cite{ier15:fangwang} introduced such a model, with more general partially-naive quasi-hyperbolic discounting, and used it, under exclusion restrictions on utility, to study present bias in mammography decisions. \cite{selfcontrolchan17} and \cite{ucb19:mahajanetal} used similar models and exclusion restrictions to analyze dynamically-inconsistent time preferences in welfare benefit choices and the demand for insecticide-treated nets. 

We first study a finite-horizon model, as in much of the literature \citep[including, e.g.][]{ucb19:mahajanetal}. For this model, a standard backward recursion argument establishes that a unique solution (up to the resolution of ties) exists. We then extend our analysis to stationary, infinite horizon models, as in \citet{ier15:fangwang}, which are useful in applications that have no clear decision horizon. We use the specific econometric structure on our infinite-horizon problem to prove that it always admits a pure-strategy solution. We show that the identification analysis of both models can be reduced to finding the discount factors $(\beta,\delta)$ that solve a system of polynomial equations. This system has one economically-interpretable equation for each exclusion restriction, with coefficients determined by the choice and transition probabilities that can be directly estimated from the data. We first focus on the case in which we have two exclusion restrictions, which give two equations to identify $\beta$ and $\delta$. We show that, in this bivariate case, the number of discount factors $(\beta,\delta)$ in the identified set is finite and bounded above by known features of the data, notably the time horizon and the number of states. In turn, each pair $(\beta,\delta)$ of discount factors in the identified set corresponds to a unique (nonparametric) utility function that rationalizes the data. 
 
Our approach leverages the assumption that agents are sophisticated and rationally foresee that they will be present biased in the future. It does not readily extend to the (partially) naive case.  \citet{ucb19:mahajanetal} showed point identification for partially-naive time preferences in a three-period model under a particular set of exclusion restrictions. \citet{daljordetal19} established point identification of a fully nonparametric, time-separable discount function in a  terminal action problem.  We demonstrate that, for a more general set of exclusion restrictions than \citeauthor{ucb19:mahajanetal}'s and more general choice structures and time horizons, but sophisticated time preferences,  the number of discount functions in the identified set depends on the number of choice periods (with finite horizons) or states (with infinite horizons).\footnote{\citeauthor{ier15:fangwang} proposed a proof of identification of partially-naive quasi-hyperbolic time preferences under similar exclusion restrictions to the ones we consider in this paper. \cite{ier20:abbringdaljord} showed that \cite{ier15:fangwang}'s main identification claim is void--- that  it has no implications for identification of the dynamic discrete choice model--- and that its main proof of identification is incorrect and incomplete. \citeauthor{selfcontrolchan17} builds on \citeauthor{ier15:fangwang}'s intuition. We emphasize that we do not believe that the incorrect results in \citeauthor{ier15:fangwang} invalidates the results in \citeauthor{selfcontrolchan17}. On the contrary, we think our results confirm that the model in \citeauthor{selfcontrolchan17}, and possibly also the one in \citeauthor{ier15:fangwang}, are formally identified.}  
Along the way, we provide a way to concentrate the empirical analysis of such models with nonparametric utility on the present-bias and standard discount factors.

After showing that the quasi-hyperbolic discount function parameters $\beta$ and $\delta$ are formally identified, we note that though their product $\beta \delta$ can be precisely estimated in finite samples, the parameters are likely to be estimated individually at comparably low levels of precision. We show that the poor finite-sample performance follows from the way the parameters $\beta$ and $\delta$ enter largely interchangeably in the identifying moment conditions. We illustrate this point with simulations of a simple three-period model with two exclusion restrictions, which we know to be point identified. We then demonstrate that it also holds in an infinite horizon application to the demand for margarine, which comes with many natural exclusion restrictions.

Our results suggest to more closely follow the experimental literature and develop an identification strategy explicitly around the concept of preference reversals.  In \cites{thaler81} classic example, subjects who prefer one apple today to two apples tomorrow tend to prefer two apples one year and one day from now to one apple one year from now. Such preference reversals are the defining feature of present-biased time preferences. In empirical work, demand for commitment devices has been viewed as a strategic response to anticipated preference reversals and has been taken as evidence of sophisticated present-bias \citep[e.g.][]{malmendierdellavigna06}. A sophisticated agent may want to lock in her savings to avoid excessive spending by future present-biased selves, that is, to restrict her future choice sets without receiving a current period pay-off. Demonstrated willingness to restrict one's future choice set may be a more promising approach. It has yet not, to our knowledge, been used formally as part of an identification strategy for dynamic discrete choice models.

\section{Model}
\label{sec:model}

We first study a finite horizon dynamic discrete choice model in which agents may suffer from sophisticated present-bias. This model is similar to \cites{ier15:fangwang}, but is nonstationary, and does not allow for partial naivity. In Section \ref{sec:infinitehorizon}, we study a stationary, infinite horizon version like \citeauthor{ier15:fangwang}'s.

\subsection{Primitives}

Time is indexed by $t = 1, \hdots, T$; with $T < \infty$.  In each period $t$, the agent chooses an action $d_t$ from a finite set $\mathcal{D} = \{1, \hdots, K\}$. Prior to making this choice, the agent draws and observes vectors of state variables $x_t$ and $\epsilon_t = \{\epsilon_{1,t}, \hdots, \epsilon_{K,t}\}$. The observable (to the econometrician) states $x_t$ have finite support $\mathcal{X}$ and evolve as a controlled (by $d_t$) first order Markov process. For notational simplicity only, we take this process to be stationary, with Markov transition distribution $Q_k$ if $k\in{\cal D}$ is chosen. The utility shocks $\epsilon_{k,t}$ are independent from $x_t$ and prior states and choices, over time, and across choices, and have type 1 extreme value distributions.\footnote{Our analysis straightforwardly extends to the case in which the vectors $\epsilon_t$ are independent over time, with {\em known} continuous distributions $G_t$ on a common support $\mathbb{R}^K$. Note that the exact choice of $G_t$, for $t=1,\ldots,T$, does not impose testable restrictions on the type of data that we assume are available in this paper.}

If, in period $t$ and state $x$, the agent chooses $k$, she collects a flow of utility $u_{k,t}(x) + \epsilon_{k,t}$. We normalize $u_{K,t}(x) = 0$ for all $t \in 1, \hdots, T$ and $x \in \mathcal{X}$. This normalization is substantive, but is standard in the literature and cannot be rejected by the type of observational data on choices and states that we will assume available in this paper.\footnote{The way $u_{K,t}$ is normalized affects the model's implied behavioural responses to many, but not all, counterfactual interventions \citep[e.g.][]{res14:noretstang,qme14:aguirregabirisuzuki,qe18:kalouptsidietal}.} 

The agent's discount function has two parameters: a non-negative and finite \emph{standard discount factor} $\delta$ and a \emph{present-bias parameter} $\beta \in (0,1]$.  Since the horizon is finite, we do not require that the discount factor $\delta$ is smaller than one. If $\beta = 1$, the model reduces to one with standard geometric discounting. The present-bias parameter is bounded away from zero to distinguish present-bias from myopia.

\subsection{Choices}

Choices in dynamic discrete choice models are regulated by value functions. Since present-biased time preferences are time inconsistent, these value functions do not follow from a standard dynamic program. It is common to think about the values as summarizing the pay-offs to players in a Stackelberg-like game played between selves in different time periods  \citep[e.g.][]{elster85}. 

Let $\tilde{\sigma}_t: \mathcal{X}\times \mathbb{R}^K \rightarrow \mathcal{D}$ be an arbitrary choice strategy and $\mathbf{\tilde{\sigma}}_t = \{\tilde{\sigma}_\tau\}^T_{\tau = t}$ an arbitrary strategy profile.  The agent's \emph{current choice specific value function}, which regulates the choices, is
\begin{align}
\label{eq:w}
w_{k,t}(x;\tilde{\mathbf{\sigma}}_{t+1}) & = u_{k,t}(x) + \beta \delta \int v_{t+1}(x';\tilde{\mathbf{\sigma}}_{t+1})dQ_k(x'|x)
\end{align}
for $t<T$, with terminal value $w_{k,T}(x)  = u_{k,T}(x)$. The agent trades off current utility versus future values by factor $\beta \delta$, but the stream of all future utilities are discounted geometrically by factor $\delta$ according to the \emph{perceived long run value function}, which equals 
\begin{align}\label{eq:perceivedvalue}
v_{t+1}(x;\tilde{\mathbf{\sigma}}_{t+1}) = \mathbb{E}_{\epsilon_{t+1}}&\bigg[u_{\tilde{\sigma}_{t+1}(x, \epsilon_{t+1}),t+1}(x) + 
 \epsilon_{\tilde{\sigma}_{t+1}(x, \epsilon_{t+1}),t+1}\nonumber\\ &+ \delta \int v_{t+2}(x';\tilde{\mathbf{\sigma}}_{t+2})dQ_{\tilde{\sigma}_{t+1}(x, \epsilon_{t+1})}(x'|x) \bigg]
\end{align}
for $t+1<T$, with terminal value $v_{T}(x;\tilde{\mathbf{\sigma}}_{T}) = \mathbb{E}_{\epsilon_T}\left[u_{\tilde{\sigma}_{T}(x, \epsilon_T),T}(x) +  \epsilon_{\tilde{\sigma}_{T}(x, \epsilon_T),T}\right]$. 

The perceived long run value depends on the current self's perceptions of its future selves' strategies $\tilde{\sigma}_{t+1}$.  At the time of decision, each of these future selves have present-biased preferences which are in conflict with the current self's time consistent long run time preferences. 

Since the agent is sophisticated, her perceptions of her future strategies are correct in equilibrium. Thus, in a sophisticated intrapersonal equilibrium, her selves use a \emph{perception perfect strategy}  \citep{odonoghuerabin99}, which is a strategy profile $\mathbf{\sigma}^*_1$ such that each $\sigma^*_t$ is a best response to her perceived future strategy profile ${\mathbf{\sigma}}^*_{t+1}$:
\begin{align}
\label{eq:ppe}
\sigma^*_t(x, \epsilon_t) & \in  \arg \max_{k \in \mathcal{D}} \{w_{k,t}(x;{\mathbf{\sigma}}^*_{t+1}) + \epsilon_{k,t} \}.
\end{align}
Here, $w_{k,T}(x;{\mathbf{\sigma}}^*_{T+1})$ should be read as $w_{k,T}(x)$.

It is easy to show, by backward induction from time $T$, that a perception perfect strategy exists and is unique, up to the resolution of ties in the decision in \eqref{eq:ppe}. Because $\epsilon_t$ is continuously distributed, the implied equilibrium probabilities 
\begin{align}
\label{eq:ccp}
\cp_{k,t}(x;\mathbf{\sigma}^*_t)&= \mathbb{E}_{\epsilon_t} [\mathbbm{1}\{\sigma^*_t(x, \epsilon_t) = k\}]
\end{align}
that the agent chooses $k\in{\cal D}$ in state $x\in{\cal X}$ are unique.

\section{Identification}\label{sec:identification}

Suppose that the data provide the state transition probabilities $Q_1,\ldots,Q_K$; and the conditional choice probabilities $\cp_{k,t}(x)=\Pr(d_t=k | x_t=x)$ for all $k\in{\cal D}$; $t=1,.\ldots,T$; and $x\in{\cal X}$.\footnote{\label{fn:dgp} In empirical applications, one typically observes draws from the joint process $\{(x_t,d_t); t=1,\ldots,T\}$. Under our model's assumptions, this process is first order Markovian, with a transition distribution that is fully characterized by the state transition and conditional choice probabilities. Of course, this Markov property can be tested and its rejection may point to, for example, persistent unobserved heterogeneity. Such unobserved heterogeneity can be handled in the usual way. Here, we concentrate on what can be learned from dynamic discrete choice data once individual state transition and conditional choice probabilities are identified.} This section studies the extent to which these data uniquely determine the model primitives ${Q}_1,\ldots,{Q}_K$; $\beta$; $\delta$; and $u_{1,t},\ldots,u_{K-1,t}$;  $t=1,\ldots,T$.  Clearly, the observed state transitions directly identify ${Q}_k$ and thus, because they were assumed rational, the agent's expectations. We therefore focus on the identification of the utility functions $u_{k,t}$ and the discount parameters $\beta$ and  $\delta$ from the conditional choice probabilities for given $Q_1,\ldots,Q_K$. To this end, we assume that the observed conditional choice probabilities are generated by a perception perfect strategy of the model:
\begin{equation}
\label{eq:p}
\cp_{k,t}(x)=\cp_{k,t}(x;\mathbf{\sigma}^*_t)\text{ for all }k\in{\cal D};~ t=1,.\ldots,T;\text{ and }x\in{\cal X}.
\end{equation}

\subsection{Basic Results}

The choice probabilities only depend on the primitives through the value contrasts $w_{k,t}(x;\mathbf{\sigma}^*_{t+1}) - w_{K,t}(x;\mathbf{\sigma}^*_{t+1})$. In particular, \eqref{eq:p} implies that\footnote{The functional form of the mapping between value contrasts and choice probabilities is specific to the assumption that the $\epsilon_{k,t}$ have independent type 1 extreme value distributions, but the results given here extend to general known $G_t$.}
\begin{align}\label{eq:HM}
\ln\left(\frac{\cp_{k,t}(x)}{\cp_{K,t}(x)}\right) & = w_{k,t}(x;\mathbf{\sigma}^*_{t+1}) - w_{K,t}(x;\mathbf{\sigma}^*_{t+1})
\end{align}
for all $k\in D/\{K\}$;  $t = 1, \hdots, T$; and $x\in{\cal X}$. With the restriction that the choice probabilities add up to one over choices,  \eqref{eq:HM} gives 
\begin{align}\label{eq:preMcFaddenSurplus}
-\ln (\cp_{K,t}(x))=\ln\left(\sum_{k\in{\cal D}}\exp\left[w_{k,t}(x;\mathbf{\sigma}^*_{t+1}) - w_{K,t}(x;\mathbf{\sigma}^*_{t+1})\right]\right).
\end{align}
With $-\ln (\cp_{K,t}(x))$ in hand, \eqref{eq:HM} determines $\cp_{k,t}(x)$ from the value contrasts. 

Conversely, as in the case without present-bias \citep{res93:hotzmiller}, using \eqref{eq:HM}, the current choice specific value contrasts can be uniquely recovered from the observed choice probabilities. Altogether, this implies that we can focus our identification analysis on the question to what extent the discount parameters and utilities are uniquely determined from the value contrasts $w_{k,t}(x;\mathbf{\sigma}^*_{t+1}) - w_{K,t}(x;\mathbf{\sigma}^*_{t+1})$, for given $Q_k$.

It is well known that the dynamic discrete choice model with geometric discounting ($\beta=1$) is not identified (\citealp{nh94:rust}, Lemma 3.3, and  \citealp{ecta02:magnacthesmar}, Proposition 2). The underidentification carries over to its generalization with present-bias. Specifically, the following version of \cites{ecta02:magnacthesmar} Proposition 2 holds.
\begin{theorem}[{\bf Nonidentification}] 
\label{th:nonident}
For given ${Q}_1,\ldots,{Q}_K$; $\beta$; $\delta$; and $\cp_{k,t}(x)$; $k\in{\cal D}$; $t=1,.\ldots,T$; and $x\in{\cal X}$; there exists unique utility functions $u_{1,t},\ldots,u_{K-1,t}$; $t=1,\ldots,T$; such that \eqref{eq:w}--\eqref{eq:p} hold.
\end{theorem}
\begin{proof}
Using \eqref{eq:HM}, $\cp_{k,T}$, $k\in{\cal D}$, gives the unique $w_{k,T}-w_{K,T}$, $k\in{\cal D}$, that are consistent with \eqref{eq:p}. Using the terminal condition of \eqref{eq:w} and the normalization $u_{K,T}=0$, this gives $w_{k,T}=u_{k,T}$, $k\in{\cal D}$. The strategy $\sigma^*_T$ follows up to $\epsilon$-almost sure equivalence from \eqref{eq:ppe}. Finally, $v_T$ follows from \eqref{eq:perceivedvalue}.

Next, iterate the following argument for $t=T-1,\ldots,1$. Suppose that we have constructed unique $u_{k,t+1}$, $k\in{\cal D}$, unique $v_{t+1}$, and unique (up to $\epsilon$-almost sure equivalence) $\mathbf{\sigma}^*_{t+1}=(\sigma^*_{t+1},\ldots,\sigma^*_T)$ consistent with  \eqref{eq:w}, \eqref{eq:perceivedvalue}, \eqref{eq:ppe}, and \eqref{eq:p} and the choice probabilities. For each $x\in{\cal X}$, using \eqref{eq:HM}, $\cp_{k,t}(x)$, $k\in{\cal D}$, gives the unique $w_{k,t}(x;\mathbf{\sigma}^*_{t+1})-w_{K,t}(x;\mathbf{\sigma}^*_{t+1})$, $k\in{\cal D}$, that are consistent with \eqref{eq:p}. Using \eqref{eq:w},  
\begin{align}\label{eq:Dw}
&w_{k,t}(x;\mathbf{\sigma}^*_{t+1})-w_{K,t}(x;\mathbf{\sigma}^*_{t+1})\\
&= u_{k,t}(x) - u_{K,t}(x)+ \beta \delta \int v_{t+1}(x';\mathbf{\sigma}^*_{t+1})\left[dQ_k(x'|x)-dQ_K(x'|x)\right],~k\in{\cal D}.\nonumber
\end{align}
Because the last term in the right hand side of \eqref{eq:Dw} is known at this point and $u_{K,t}(x)$ is normalized to zero, this determines $u_{k,t}$, $k\in{\cal D}$.  The strategy $\sigma^*_t$ follows up to $\epsilon$-almost sure equivalence from \eqref{eq:ppe}. Finally, $v_t$ follows from \eqref{eq:perceivedvalue}.
\end{proof}

Theorem \ref{th:nonident} implies that $\beta$ and $\delta$ can only be identified if further data are available or additional assumptions are made. In this paper, we explore identification under exclusion restrictions on the utility functions. Our analysis focuses on the identification of $\beta$ and $\delta$. Theorem \ref{th:nonident} shows that, once $\beta$ and $\delta$ are identified, unique utility functions can be found that rationalize the choice data. 

\subsection{Concentrating identification on the discount factors}

Because ${\cal X}$ is finite--- say it has $J$ elements--- it is convenient to express expectations in matrix notation. To this end, let $\mathbf{v}_t(\mathbf{\sigma}^*_t)$ be a $J\times 1$ vector that stacks the values of $v_t(x;\mathbf{\sigma}^*_t)$, $x\in{\cal X}$, and $\mathbf{Q}_k(x)$ a $1\times J$ vector that stacks the values of $Q_k(x'|x)$, $x'\in{\cal X}$, in corresponding order. 
Then, \eqref{eq:HM} and \eqref{eq:Dw}, with the normalization $u_{K,t}(x)=0$, give
\begin{align}
\label{eq:momentraw}
\ln\left(\frac{\cp_{k,t}(x)}{\cp_{K,t}(x)}\right)  = u_{k,t}(x) + \beta \delta\left[\mathbf{Q}_k(x)-\mathbf{Q}_K(x)\right]\mathbf{v}_{t+1}(\mathbf{\sigma}^*_{t+1}).
\end{align} 
Recall that, given the transition distributions $Q_k$,  \eqref{eq:momentraw} contains all information in the choice probabilities about the model's primitives. 

We will concentrate the identification analysis on the discount factors by controlling the current period utility $u_{k,t}(x)$ in the right hand side of \eqref{eq:Dw} with exclusion restrictions and expressing the continuation value in terms of the discount factors and data only. As $\mathbf{Q}_k(x)$ and $\mathbf{Q}_K(x)$ are data, this only requires that we express the perceived long run values $\mathbf{v}_{t+1}(\mathbf{\sigma}^*_{t+1})$ in terms of the discount factors and data. To this end, first substitute \eqref{eq:w} and \eqref{eq:ppe} into \eqref{eq:perceivedvalue} to get
\begin{align}
\label{eq:vw}
&v_{t+1}(x;\mathbf{\sigma}^*_{t+1})\\
&= \mathbb{E}_{\epsilon_{t+1}} \left[\max_{k\in{\cal D}}\left\{w_{k,t+1}(x;\mathbf{\sigma}^*_{t+2}) + \epsilon_{j,t+1}\right\}+\delta(1-\beta) \mathbf{Q}_{\sigma^*_{t+1}(x, \epsilon_{t+1})}(x)\mathbf{v}_{t+2}(\mathbf{\sigma}^*_{t+2})\right]. \nonumber
\end{align}
Next, as we can express the value contrast $w_{k,t+1}-w_{K,t+1}$ in terms of data using \eqref{eq:HM}, we substract $w_{K,t+1}(x;\mathbf{\sigma}^*_{t+2})$ from the first term in the right hand side of \eqref{eq:vw} and add it to the second term, which gives
\begin{flalign}\label{eq:valuetodata}
v_{t+1}&(x;\mathbf{\sigma}^*_{t+1})\\
&= m_{t+1}(x) +  w_{K,t+1}(x;\mathbf{\sigma}^*_{t+2})  + \delta(1-\beta) \mathbb{E}_{\epsilon_{t+1}} \left[\mathbf{Q}_{\sigma^*_{t+1}(x, \epsilon_{t+1})}(x)\mathbf{v}_{t+2}(\mathbf{\sigma}^*_{t+2}) \right],\nonumber
\end{flalign}
where
\begin{align}
\label{eq:m}
m_{t+1}(x) & = \mathbb{E}_{\epsilon_{t+1}} \left[\max_{k \in \mathcal{D}} \left\{w_{k,t+1}(x;\mathbf{\sigma}^*_{t+2}) - w_{K,t+1}(x;\mathbf{\sigma}^*_{t+2}) + \epsilon_{k,t+1} \right\}\right]
\end{align}
is the McFadden surplus (before observing $\epsilon_{t+1}$) for the choice among $k\in{\cal D}$ with utilities  $w_{k,t+1}(x;\mathbf{\sigma}^*_{t+2}) - w_{K,t+1}(x;\mathbf{\sigma}^*_{t+2}) + \epsilon_{k,t+1}$. Under our assumption that $\epsilon_{t+1}$ is extreme value distributed, the right-hand side of \eqref{eq:m} reduces to the right-hand side of \eqref{eq:preMcFaddenSurplus}, so that $m_{t+1}(x) =  -\ln(\cp_{K,t+1}(x))$ is known from the choice data.\footnote{More generally, given $G$, $m_{t+1}(x)$ is a known function of $w_{k,t+1}(x;\mathbf{\sigma}^*_{t+2}) - w_{K,t+1}(x;\mathbf{\sigma}^*_{t+2})$, $k \in{\cal D}$, and thus, using \eqref{eq:HM}, of the choice probabilities \citep{arcidiaconomiller11}.}
The term $w_{K,t+1}(x;\mathbf{\sigma}^*_{t+2})$ can be expressed recursively as
\begin{align}\label{eq:wK}
w_{K,t+1}(x;\mathbf{\sigma}^*_{t+2}) & = \beta \delta \mathbf{Q}_{K}(x) \mathbf{v}_{t+2}(\mathbf{\sigma}^*_{t+2}).
\end{align}
Finally, as the expectation over $\epsilon_{t+1}$ in the right hand side of \eqref{eq:valuetodata} is effectively an expectation over implied actions $\sigma^*_{t+1}(x, \epsilon_{t+1})$, it can be expressed in terms of the observed choice probabilities using \eqref{eq:ppe}:
\begin{equation}
\label{eq:expectation}
\mathbb{E}_{\epsilon_{t+1}} \left[\mathbf{Q}_{\sigma^*_{t+1}(x, \epsilon_{t+1})}(x)\mathbf{v}_{t+2}(\mathbf{\sigma}^*_{t+2}) \right]  = 
\sum_{k\in \mathcal{D}}\cp_{k,t+1}(x)\mathbf{Q}_k(x)\mathbf{v}_{t+2}(\mathbf{\sigma}^*_{t+2}).
\end{equation}
Substituting \eqref{eq:wK} and \eqref{eq:expectation} into \eqref{eq:valuetodata} gives
\begin{align}
\label{eq:vx}
v_{t+1}(x;\mathbf{\sigma}^*_{t+1}) & = m_{t+1}(x) + \delta \left[\beta \mathbf{Q}_K(x)+ (1-\beta)\overline{\mathbf{Q}}_{t+1}(x)\right] \mathbf{v}_{t+2}(\mathbf{\sigma}^*_{t+2})
\end{align}
where $\overline{\mathbf{Q}}_{t+1}(x) = \sum_{k\in \mathcal{D}}\cp_{k,t+1}(x)\mathbf{Q}_k(x)$ is the expected state transition probability distribution under strategy $\sigma^*_{t+1}$ in state $x$. This mixture represents an expectation over how the choices of present-biased future selves control future state transitions, choices which are in conflict with the current self's long term preferences.  

Define the $J\times J$ matrix of probability mixtures
\begin{align}\label{eq:qpb}
\mathbf{Q}^{pb}_t(\beta) = \beta \mathbf{Q}_K + (1-\beta)\overline{\mathbf{Q}}_t,
\end{align}
where $\overline{\mathbf{Q}}_t$ stacks $\overline{\mathbf{Q}}_t(x)$ and $\mathbf{Q}_K$ stacks $\mathbf{Q}_K(x)$. Then, we can write \eqref{eq:vx} as a recursive expression for $\mathbf{v}_{t+1}(\mathbf{\sigma}^*_{t+1})$ in vector notation:
\begin{align*}
\mathbf{v}_{t+1}(\mathbf{\sigma}^*_{t+1}) & = \mathbf{m}_{t+1} + \delta \mathbf{Q}^{pb}_{t+1}(\beta) \mathbf{v}_{t+2}(\mathbf{\sigma}^*_{t+2}).
\end{align*} 
Completing the recursion until the end of time $T$ expresses  
\begin{align}\label{eq:vstar}
\mathbf{v}_{t+1}(\mathbf{\sigma}^*_{t+1})=\mathbf{m}_{t+1} + \sum^T_{\tau = t+2}\delta ^{\tau-t-1}\left(\prod^{\tau-1}_{r = t+1}\mathbf{Q}^{pb}_r(\beta)\right)\mathbf{m}_{\tau},
\end{align}
in terms of the discount factors and data only. Substituting \eqref{eq:vstar} into \eqref{eq:momentraw} gives
\begin{align}
\label{eq:momentdata}
\ln\left(\frac{\cp_{k,t}(x)}{\cp_{K,t}(x)}\right)  = &u_{k,t}(x) +\\& \beta \delta\left[\mathbf{Q}_k(x)-\mathbf{Q}_K(x)\right]\left[ \mathbf{m}_{t+1} + \sum^T_{\tau = t+2}\delta ^{\tau-t-1}\left(\prod^{\tau-1}_{r = t+1}\mathbf{Q}^{pb}_r(\beta)\right)\mathbf{m}_{\tau}\right].\nonumber
\end{align} 
The log choice probability ratio in the left hand side of  \eqref{eq:momentdata} measures the observed propensity to choose $k$ over $K$ in state $x$. The right hand side of  \eqref{eq:momentdata} explains this observed propensity by the current period's utility difference  $u_{k,t}(x)-u_{K,t}(x)=u_{k,t}(x)$ and a difference in continuation values, which is a polynomial in $\beta$ and $\delta$ with coefficients that are fully determined by the choice and transition data.  We study identification from variation in these continuation values, under exclusion restrictions on primitive utility that control the effects of variation in the current period's utility. This formalizes the common intuition that holding current period utilities constant, current choice responses to variation in future values are informative about time preferences.

\subsection{Exclusion restrictions}

Our identification argument holds for exclusion restrictions on utilities between pairs of time periods, choices, states, or any combinations of the three. To simplify the exposition, we however focus on exclusion restrictions on utilities from one given choice between pairs of states. We primarily focus on the case in which we have two such exclusion restrictions, which is the minimum needed to identify the two unknown discount factors, $\beta$ and $\delta$. In applications, intuition for exclusion restrictions would typically deliver a {\em variable} that affects continuation values, but not the current period's utility. Such a excluded variable would typically imply more than two exclusion restrictions on states, which would further restrict the identified set of discount factors.  

So, consider two exclusion restrictions on utility from choice $k\in{\cal D}/\{K\}$, indexed by $a$ and $b$. The first requires that
\begin{equation}
\label{eq:exclusionrestrictionA}
 u_{k,t_a}( x_{a,1}) = u_{k,t_a}( x_{a,2}) 
\end{equation}
at time $t_a<T-1$, for  states $x_{a,1}, x_{a,2} \in \mathcal{X}$ such that $x_{a,1}\neq x_{a,2}$. The second sets 
\begin{equation}
\label{eq:exclusionrestrictionB}
 u_{k,t_b}( x_{b,1}) =  u_{k, t_b}(x_{b,2}) 
\end{equation}
at time $t_b\leq t_a<T-1$, for  states $x_{b,1}, x_{b,2} \in \mathcal{X}$ such that $x_{b,1}\neq x_{b,2}$. Evaluating \eqref{eq:momentdata} at choice $k$ and time $t_a$, differencing between states $x_{a,1}$ and $x_{a,2}$, and using \eqref{eq:exclusionrestrictionA} gives 
\begin{align}
		\ln\left(\frac{\cp_{k,t_a}(x_{a,1})}{\cp_{K,t_a}(x_{a,1})}\right) - \ln&\left(\frac{\cp_{k,t_a}( x_{a,2})}{\cp_{K,t_a}(x_{a,2})}\right) = \nonumber \\
		 & \beta \delta \left[\mathbf{Q}_{k}(x_{a,1})-\mathbf{Q}_K(x_{a,1}) - 
		 \mathbf{Q}_{k}(x_{a,2})+\mathbf{Q}_K(x_{a,2})\right]		
		 \label{eq:systemofmomentsA}  \\
		 &~~~~~~\times
		 \left[\mathbf{m}_{t_a+1} + \sum^T_{\tau = t_a+2}\delta ^{\tau-t_a-1}\left(\Pi^{\tau-1}_{r = t_a+1}\mathbf{Q}^{pb}_r(\beta)\right)\mathbf{m}_{\tau} \right].\nonumber 
\end{align}
Similarly, \eqref{eq:momentdata}  and \eqref{eq:exclusionrestrictionB} give
\begin{align}
\ln\left(\frac{\cp_{k,t_b}(x_{b,1})}{\cp_{K,t_b}(x_{b,1})}\right) - \ln&\left(\frac{\cp_{k,t_b}(x_{b,2})}{\cp_{K,t_b}(x_{b,2})} \right)  =  \nonumber\\
		& \beta \delta \left[\mathbf{Q}_{k}(x_{b,1})-\mathbf{Q}_K(x_{b,1})   - 
		 \mathbf{Q}_{k}(x_{b,2})+\mathbf{Q}_K(x_{b,2})\right]\label{eq:systemofmomentsB}\\
		 &~~~~~~\times\left[\mathbf{m}_{t_b+1} + \sum^T_{\tau = t_b+2}\delta ^{\tau-t_b-1}\left(\Pi^{\tau-1}_{r = t_b+1}\mathbf{Q}^{pb}_r(\beta)\right)\mathbf{m}_{\tau} \right].\nonumber
\end{align}
The left hand sides of \eqref{eq:systemofmomentsA} and \eqref{eq:systemofmomentsB} are scalars that are known from the data. They reflect choice differences between the states that figure in each exclusion restrictions and that are known from the data. The right hand sides of \eqref{eq:systemofmomentsA} and \eqref{eq:systemofmomentsB} are the model's implications for these same choice responses at discount factors $\beta$ and $\delta$. They are bivariate polynomials in $\beta$ and $\delta$ of order $T-t_a$ and $T-t_b$, respectively, with coefficients that are fully determined by the data and do not depend on the unknown utility functions $u_{k,t}$. By Theorem \ref{th:nonident}, for any pair of discount factors $(\beta, \delta)$ that solves  \eqref{eq:systemofmomentsA} and \eqref{eq:systemofmomentsB}, unique utilities $u_{k,t}$ can be found that rationalize the choice data and, in particular, solve \eqref{eq:HM}. Therefore, the identified set can be characterized by finding all discount factors  $(\beta,\delta)$ that solve the bivariate polynomial equations \eqref{eq:systemofmomentsA} and \eqref{eq:systemofmomentsB}. Solving systems of bivariate polynomial equations is a well-understood problem where we can draw on standard results from algebraic geometry. Before we use these results to formally characterize the identified set, we first give a simple three-period example.

\begin{example}
\label{ex:threeperiod}
Let $T=3$. Suppose that the exclusion restrictions \eqref{eq:exclusionrestrictionA}
and \eqref{eq:exclusionrestrictionB} hold with $t_a=2$, $t_b=1$, $x_{a,1}=x_{a,2}=x_a$, and $x_{b,1}=x_{b,2}=x_b$, with $x_a\neq x_b$. The first exclusion restriction, $u_{k,2}(x_a)=u_{k,2}(x_b)$, gives
\begin{align}\label{eq:exmom1}
\ln\left(\frac{\cp_{k,2}(x_{a})}{\cp_{K,2}(x_{a})}\right) -\ln&\left(\frac{\cp_{k,2}(x_{b})}{\cp_{K,2}(x_{b})}\right) \nonumber\\ &= \beta \delta \left[\mathbf{Q}_{k,2}(x_{a})-\mathbf{Q}_{K,2}(x_{a})-\mathbf{Q}_{k,2}(x_{b})+\mathbf{Q}_{K,2}(x_{b})\right] \mathbf{m}_{3}.
\end{align}
Provided that 
\begin{equation}
\label{eq:rank}
\left[\mathbf{Q}_{k,2}(x_{a})-\mathbf{Q}_{K,2}(x_{a})-\mathbf{Q}_{k,2}(x_{b})+\mathbf{Q}_{K,2}(x_{b})\right] \mathbf{m}_{3}\neq 0,
\end{equation} 
this equation identifies $\beta\delta$. This is standard and  intuitive. The product $\beta\delta$ is the factor with which the agent in period $2$ discounts the expected utility in the terminal period $3$, which itself does not depend on $\beta$ or $\delta$. The exclusion restriction ensures that the current utility from choice $k$ in period $2$ does not vary between states $x_a$ and $x_b$, so that the choice response in the left hand side of \eqref{eq:exmom1} only reflects how much the agent cares about expected utility in the third period, $\beta\delta$, and how much that expected utility varies between both states, which is nonzero by \eqref{eq:rank}.

Additional exclusion restrictions on utility in period 2 leads to conditions similar to \eqref{eq:exmom1} that identify $\beta\delta$, but these can not identify $\beta$ and $\delta$ separately.  For separate identification of  $\beta$ and $\delta$, we need data on at least three periods. The second exclusion restriction, $u_{k,1}(x_a)=u_{k,1}(x_b)$, allows us to bring in data on the first period and gives the moment condition
\begin{align}\label{eq:exmom2}
\ln\left(\frac{\cp_{k,1}(x_{a})}{\cp_{K,1}(x_{a})}\right) -\ln\left(\frac{\cp_{k,1}(x_{b})}{\cp_{K,1}(x_{b})}\right) &= \beta \delta \left[\mathbf{Q}_{k,1}(x_{a})-\mathbf{Q}_{K,1}(x_{a})-\mathbf{Q}_{k,1}(x_{b})+\mathbf{Q}_{K,1}(x_{b})\right] \nonumber\\
&~~\times\left[\mathbf{m}_{2} +\left(\beta\delta\mathbf{Q}_{K,2} +\left(1-\beta\right)\delta \bar{\mathbf{Q}}_2\right)\mathbf{m}_3\right].
\end{align}
Because $\beta\delta$ is already identified from \eqref{eq:exmom1} under \eqref{eq:rank}, \eqref{eq:exmom2} identifies $\beta$ and $\delta$ if
\begin{align}\label{eq:rankcond2}
\left[\mathbf{Q}_{k,1}(x_{a})-\mathbf{Q}_{K,1}(x_{a})-\mathbf{Q}_{k,1}(x_{b})+\mathbf{Q}_{K,1}(x_{b})\right] \bar{\mathbf{Q}}_2\mathbf{m}_3\neq 0.
\end{align}
\end{example}
In this simple example, the solutions are easy to find. We however need a more systematic approach for the general case. We can analyze the solutions to the moment conditions using standard results from algebraic geometry. Our exposition builds on \cite{coxetal15}. The moment conditions  \eqref{eq:exmom1} and \eqref{eq:exmom2} can be expressed in terms of their \emph{Sylvester matrix}. To construct the Sylvester matrix for an arbitrary bivariate system in $\beta$ and $\delta$, we  choose either $\beta$ or $\delta$ as a base and treat it as constant without loss of generalization. Suppose the two moment conditions, taking $\delta$ to be the base, are polynomials $q_a$ and $q_b$ in $\beta$, of order  $l$ and $m$, respectively, with coefficients $a_0, \hdots, a_l$ and $b_0, \hdots, b_m$, respectively. The Sylvester matrix with base $\delta$ is
\begin{align*}
Syl(q_a(\beta, \delta), q_b(\beta, \delta))_\delta& =
\left(
\begin{array}{cccc cccc}
a_{0} & & & & b_0 &  & & \\
a_{1} &a_0 & & & b_1& b_0 & & \\
a_{2} & a_1& \ddots& &  b_2 & b_1  & \ddots & \\
\vdots & & & a_0 & \vdots &  & & b_0 \\
a_{l} & \vdots & & & b_m & \vdots & & \\
& a_l& & \vdots &  & b_m & & \vdots\\
 & & \ddots & &  &  & \ddots & \\
 & & &a_l & &  & &b_m \\
\end{array}
\right),
\end{align*}
an $(l+m) \times (l+m)$ matrix, where the blanks are zeros.
The \emph{resultant}, the determinant of the Sylvester matrix, is a polynomial in the base parameter. From Corollary 4, chapter 3, of \cite{coxetal15}, each $\delta$ root of the resultant is a root of the moment conditions. This eliminates the base parameter $\delta$, and we can solve for the roots of the remaining parameter, in this case $\beta$. Though we use the resultant to show identification, it also gives an algorithm to solve for the roots of the system: for each value of $\delta$ that sets the resultant to zero, we can solve for the roots of the remaining univariate polynomial in $\gamma$. The literature shows that this algorithm may work well for low order polynomials, but becomes computationally inefficient and unstable when the order exceed e.g. 30. 
\\
\\
We illustrate the resultant using Example \ref{ex:threeperiod} where we reparametrize the problem to be bivariate in $\gamma = \beta \delta$ and $\delta$. Write the moment conditions in  \eqref{eq:exmom1} and \eqref{eq:exmom2} as
\begin{align}
f_a :& \,a_1 \gamma +a_0  = 0\\
f_b: & \, b_2 \gamma^2 + b_1(\delta)\gamma +b_0  = 0
\end{align}
where
\begin{align*}
a_1 & = \left[\mathbf{Q}_{k,2}(x_{a})-\mathbf{Q}_{K,2}(x_{a})-\mathbf{Q}_{k,2}(x_{b})+\mathbf{Q}_{K,2}(x_{b})\right] \mathbf{m}_{3} \\
-a_0 & = \ln\left(\frac{\cp_{k,2}(x_{a})}{\cp_{K,2}(x_{a})}\right) -\ln\left(\frac{\cp_{k,2}(x_{b})}{\cp_{K,2}(x_{b})}\right) \\
b_2 & = \left[\mathbf{Q}_{k,1}(x_{a})-\mathbf{Q}_{K,1}(x_{a})-\mathbf{Q}_{k,1}(x_{b})+\mathbf{Q}_{K,1}(x_{b})\right][\mathbf{Q}_{K,2} -\bar{\mathbf{Q}}_2]\mathbf{m}_3 \\
b_{1,1} & =  \left[\mathbf{Q}_{k,1}(x_{a})-\mathbf{Q}_{K,1}(x_{a})-\mathbf{Q}_{k,1}(x_{b})+\mathbf{Q}_{K,1}(x_{b})\right]\bar{\mathbf{Q}}_2\mathbf{m}_3 \\
b_{1,2} & =  \left[\mathbf{Q}_{k,1}(x_{a})-\mathbf{Q}_{K,1}(x_{a})-\mathbf{Q}_{k,1}(x_{b})+\mathbf{Q}_{K,1}(x_{b})\right]\mathbf{m}_{2}\\
b_1(\delta) & = \delta b_{1,1} + b_{1,2}\\
-b_0 & = \ln\left(\frac{\cp_{k,1}(x_{a})}{\cp_{K,1}(x_{a})}\right) -\ln\left(\frac{\cp_{k,1}(x_{b})}{\cp_{K,1}(x_{b})}\right) \\
\end{align*}
In this example, we get
\begin{align}
Syl(f_a(\gamma, \delta), f_b(\gamma, \delta))_\delta & = 
\left(
\begin{array}{ccc}
 a_1 &0 & b_2 \\
a_0 & a_1 & b_1(\delta)  \\
 0  & a_0& b_0\\
\end{array}
\right)
\end{align}
The resultant is the determinant $Syl(f_a(\gamma, \delta), f_b(\gamma, \delta))_\gamma$
\begin{align*}
Res(f_a(\gamma, \delta), f_b(\gamma, \delta))_\gamma  & =  a_1(a_1b_0-a_0b_1(\delta)) + a_0^2b_2 \nonumber \\
& = -a_0a_1b_1(\delta) + a_1^2b_0 + a_0^2b_2 \nonumber \\
& =-\delta a_0a_1b_{1,1} + a_0(a_1b_{1,2}+a_0b_2) +  a_1^2b_0 
\end{align*}
Setting the resultant to zero gives a unique root
\begin{align}
\delta^* & = \frac{a_0a_1 b_{1,2}+ a_1^2b_0 + a_0^2b_2}{a_0a_1b_{1,1}},
\end{align}
if it exists. Inserting $\delta^*$ in \eqref{eq:exmom1}, which is linear in $\beta$, gives the solution, the common roots of the moment conditions. 

We see that if $b_{1,1} = 0$, which violates the non-zero condition in \eqref{eq:rankcond2}, then no root $\delta^*$ exists. This rejects the model: We can not find parameters $\beta$ and $\delta$ that rationalize the data. If $a_1 = 0$, which violates the rank condition in \eqref{eq:rank}, and $a_0 = 0$, then  the resultant is zero for all values of the $\delta$. This implies that the moment conditions have a \emph{common factor}: it has an infinity of solutions and all identification is lost. The resultant generalizes the rank condition in the geometric case, where the identifying moment condition is a univariate polynomial moment condition, to the hyperbolic case, where the identifying moment conditions are a bivariate  set of polynomial moment conditions.

The example illustrates that the common factors of \eqref{eq:systemofmomentsA} and \eqref{eq:systemofmomentsB} characterize exceptions to the identification of $(\beta,\delta)$ from these moment conditions. We therefore need a definition of common factors.  Write the moment conditions  \eqref{eq:systemofmomentsA} and \eqref{eq:systemofmomentsB} as $f_a(\beta,\delta)=0$ and $f_b(\beta,\delta)=0$, respectively, with $f_a$ and $f_b$ $(T-t)$'th order polynomials. 
\begin{definition}
The polynomials $f_a$ and $f_b$ have a {\em common factor} $h$ if $f_a(\beta,\delta)=h(\beta,\delta)g_a(\beta,\delta)$ and $f_a(\beta,\delta)=h(\beta,\delta)g_a(\beta,\delta)$, with $h$ a polynomial of of order one or higher and $g_a$ and $g_b$ polynomials.
\end{definition}
 A simple example of a common factor of \eqref{eq:systemofmomentsA} and \eqref{eq:systemofmomentsB} in the case that their left hand sides are zero (no choice responses) is $h(\beta,\delta)=\delta$. 
\begin{theorem}\label{th:theorem1}[{\bf Identified set}]
Suppose that the exclusion restrictions in \eqref{eq:exclusionrestrictionA} and \eqref{eq:exclusionrestrictionB} hold and that the polynomial differences between the left and right hand side of \eqref{eq:systemofmomentsA} and \eqref{eq:systemofmomentsB} are nonzero and have no common factors. Then the identified set contains at most $(T-t_a)(T-t_b)$ elements. 
\end{theorem}

\begin{proof}
By B\'{e}zout's Theorem \citep{bezout}, the system then has no more than $(T-t_a)(T-t_b)$ zeros in $\mathbb{C}^2$, which is also an upper bound on the number of zeros on the domain of $\beta$ and $\delta$. 
\end{proof}
\noindent Bezout's theorem generalizes the fundamental theorem of algebra to multivariate polynomials, see e.g. \cite{coxetal15}. Theorem \ref{th:theorem1} does not guarantee a solution. The zero set may be empty, such as in the example above. In that case, the model is rejected.\footnote{See \cite{qe20:abbringdaljord} for a discussion of the empirical content of dynamic discrete choice models under exclusion restrictions.} 

Except for certain special cases, such as when one moment condition is a multiple of the other, we have not found obvious economic interpretations of common factors. The existence of common factors can however easily be verified on a case-by-case basis by calculating the resultant. In practice, a resultant that is everywhere close to zero signals numerical instability similar to a regression matrix with a high condition number where identification is close to lost to multicollinearity.

\section{Extension to an infinite horizon}
\label{sec:infinitehorizon}

Section \ref{sec:identification} analyzes a finite horizon decision model. We will now consider an infinite horizon, stationary version of that model. This is \cites{ier15:fangwang} model, but with sophisticated agents. We will provide a new equilibrium existence result for this model and show that Section \ref{sec:identification}'s identification results extend to it.

\subsection{Model}

Take Section \ref{sec:model}'s model, but let time be indexed by $t \in\{1, 2, \hdots\}$ and utility from choice $k$ in state $x$ at time equal $u_{k}(x) + \epsilon_{k,t}$, with $u_k$ time invariant.  Also, restrict $\delta$ to $[0,1)$ to ensure convergence of the agent's expected discounted utility (recall that $\beta \in (0,1]$). 

Because neither the model primitives nor the decision horizon change over time, the resulting decision problem is stationary. Consequently, it is natural for the agent to use a \emph{stationary Markov strategy}, say $\tilde{\sigma}: \mathcal{X}\times \mathbb{R}^K \rightarrow \mathcal{D}$, to map the current state (but not time itself) into an action in each period. If the agent uses such a strategy $\tilde\sigma$ in all future periods, her current choice-$k$ specific value in state $x$ equals
\begin{align}
\label{eq:wStatscal}
w_{k}(x;\tilde{\sigma})  = u_{k}(x) + \beta \delta \int v(x';\tilde{\sigma})dQ_k(x'|x),
\end{align}
with perceived long run value 
\begin{align}\label{eq:perceivedvalueStatscal}
v(x;\tilde{\sigma}) = \mathbb{E}_{\epsilon_{t+1}}\bigg[u_{\tilde{\sigma}(x, \epsilon_{t+1})}(x) +  \epsilon_{\tilde{\sigma}(x, \epsilon_{t+1}),t+1}+ \delta \int v(x';\tilde{\sigma})dQ_{\tilde{\sigma}(x, \epsilon_{t+1})}(x'|x) \bigg].
\end{align}
The agent uses a \emph{stationary perception perfect strategy} $\sigma^*$, which is a best response to employing the same strategy in all future periods:
\begin{align}
\label{eq:sppe}
\sigma^*(x, \epsilon) & \in \arg \max_{k \in \mathcal{D}} \{w_{k}(x;{\sigma}^*) + \epsilon_{k} \}\text{ for all }(x,\epsilon)\in{\cal D}\times\mathbb{R}. 
\end{align}
In turn, the strategy $\sigma^*$ implies stationary conditional choice probabilities
\begin{equation}
\label{eq:ccpStat}
\cp_{k}(x;\sigma^*)=\mathbb{E}_{\epsilon_t} [\mathbbm{1}\{{\sigma}^*(x, \epsilon_t) = k\}];~~k\in{\cal D},~x\in{\cal X}.
\end{equation}  

\subsection{Existence of a stationary perception perfect strategy}

We will exploit the specific econometric structure of our model, in particular the presence of additively separable and independent utility shocks, to demonstrate existence of a stationary perception perfect strategy.\footnote{\citet[][p. 570]{ier15:fangwang} claims, for a model of a partially naive agent that encompasses ours, that  ``[t]he existence of the perception-perfect strategy profile is shown in Peeters (2004) for the same class of stochastic games.''  \citet{mu04:peeters}, however, assumes a finite state space, so that its results cover neither \citeauthor{ier15:fangwang}'s model nor ours. If anything, the analysis in \citeauthor{mu04:peeters} suggests that the pure and stationary perfection perfect strategies that we focus on do {\em not} always exist: For its case with finite states, it demonstrates that mixed strategies are needed to ensure existence.}
To this end, we return to Section \ref{sec:identification}'s matrix notation, without time subscripts. From \eqref{eq:wStatscal}, the current choice-$k$ specific values under perceived future use of $\sigma^*$ equal
\begin{align}
\label{eq:wStat}
\mathbf{w}_{k}(\sigma^*) & = \mathbf{u}_{k} + \beta \delta\mathbf{Q}_k\mathbf{v}(\sigma^*).
\end{align}
Using a straightforward adaption of Section \ref{sec:identification}'s derivations, we can rewrite \eqref{eq:perceivedvalueStatscal} as
\[
\mathbf{v}({\sigma}^*) = -\ln\mathbf{\cp}_K(\sigma^*)+\mathbf{w}_K(\sigma^*)+\delta(1-\beta)\sum_{k\in{\cal D}}\mathbf{\CP}_k(\sigma^*){\mathbf{Q}_k}\mathbf{v}({\sigma}^*),
\]
where $\mathbf{\cp}_k(\sigma^*)$ is a $J\times 1$ vector that stacks the conditional choice probabilities $\cp_{k}(x;\sigma^*)$, $x\in{\cal X}$,  and $\mathbf{\CP}_k(\sigma^*)$ is a $J\times J$ diagonal matrix with the same probabilities on its diagonal.\footnote{As in Section \ref{sec:identification}, $-\ln\mathbf{\cp}_K(\sigma^*)$ is a vector of McFadden surpluses. We could denote it with $\mathbf{m}(\sigma^*)$, but it is important for the equilibrium analysis to keep its dependence on $\mathbf{\cp}_K(\sigma^*)$ explicit.}
Substituting \eqref{eq:wStat} for $k=K$ and rearranging gives
\begin{align}\label{eq:perceivedvalueStat}
\mathbf{v}({\sigma}^*) = \left[\mathbf{I}-\delta\left(\beta\mathbf{Q}_K+(1-\beta)\sum_{k\in{\cal D}}\mathbf{\CP}_k(\sigma^*){\mathbf{Q}_k}\right)\right]^{-1}\left[-\ln\mathbf{\cp}_K(\sigma^*)+\mathbf{u}_K\right],
\end{align}
with $\mathbf{I}$ a $J\times J$ identity matrix. Note that $\beta\mathbf{Q}_K+(1-\beta)\sum_{k\in{\cal D}}\mathbf{\CP}_k(\sigma^*){\mathbf{Q}_k}$ is a stochastic matrix and that $\delta\in[0,1)$, so that the inverse in the right hand side of \eqref{eq:perceivedvalueStat} exists. 

Now stack  $\mathbf{\cp}_1(\sigma^*),\ldots,\mathbf{\cp}_K(\sigma^*)$ in a $K J\times 1$ vector $\mathbf{\cp}(\sigma^*)$. Note that $\mathbf{\cp}(\sigma^*)$ takes values in ${\cal \CP}\equiv\left\{(\tilde{\mathbf{\cp}}_1,\ldots,\tilde{\mathbf{\cp}}_K)\in\mathbb{R}^{K J}\;\vline\; \tilde{\mathbf{\cp}}_1,\ldots,\tilde{\mathbf{\cp}}_K\geq 0;~\sum_{k\in{\cal D}}\tilde{\mathbf{\cp}}_k=\bm{1}\right\}$, with $\bm{1}$ a $K\times 1$ vector of ones. Substituting \eqref{eq:perceivedvalueStat} in \eqref{eq:wStat} gives 
$\mathbf{w}_k(\sigma^*)=\bm{\omega}_k(\mathbf{\cp}(\sigma^*))$,
with $\bm{\omega}_k: {\cal \CP}\rightarrow\mathbb{R}^J$ such that
\[
\bm{\omega}_k(\tilde{\mathbf{\cp}})=
\mathbf{u}_{k} + \beta \delta\mathbf{Q}_k\left[\mathbf{I}-\delta\left(\beta\mathbf{Q}_K+(1-\beta)\sum_{k\in{\cal D}}\tilde{\mathbf{\CP}}_k{\mathbf{Q}_k}\right)\right]^{-1}\left[-\ln\tilde{\mathbf{\cp}}_K+\mathbf{u}_K\right].
\]
Note that $\bm{\omega}_k$ does not depend on $\sigma^*$, so that $\mathbf{w}_k(\sigma^*)=\bm{\omega}_k(\mathbf{\cp}(\sigma^*))$ only depends on $\sigma^*$ through $\mathbf{\cp}(\sigma^*)$.
Let $\bm{\omega}_k(x;\mathbf{\cp}(\sigma^*))$ be the element of $\bm{\omega}_k(\mathbf{\cp}(\sigma^*))$ corresponding to state $x\in{\cal X}$. For $\sigma^*$ to be a stationary perception perfect strategy, it needs to be a best response to the perception that later choices are made with the probabilities in $\mathbf{\cp}(\sigma^*)$, as in \eqref{eq:sppe} with $w_k(x;\sigma^*)=\bm{\omega}_k(x;\mathbf{\cp}(\sigma^*))$. Such a strategy exists.

\begin{theorem}[{\bf Existence}]
\label{th:existence}
The infinite horizon dynamic discrete choice model with sophisticated quasi-hyperbolic discounting has a stationary perception perfect strategy.
\end{theorem}
\begin{proof}
Let $\varpi:\left(\mathbb{R}^J\right)^K\rightarrow{\cal \CP}$ map choice specific values into logit choice probabilities. That is,  for $\mathbf{w}_1,\ldots,\mathbf{w}_K\in\mathbb{R}^J$, $\varpi(\mathbf{w}_1,\ldots,\mathbf{w}_K)$ stacks 
$\exp\left[w_k(x)\right]/\sum_{l\in{\cal D}}\exp\left[w_l(x)\right]$; $x\in{\cal X}$, $k\in{\cal D}$. Then $\pi\equiv\varpi\circ(\bm{\omega}_1,\ldots,\bm{\omega}_K):{\cal \CP}\rightarrow{\cal \CP}$ maps perceived future choice probabilities to the logit choice probabilities implied by the best response to them (which, unlike the best response in \eqref{eq:sppe}  itself, are unique, because the agent is only indifferent on a set of $\epsilon$ with probability zero). The fixed points of $\pi$ are the choice probabilities $\mathbf{\cp}(\sigma^*)$ corresponding to a stationary perception perfect strategy $\sigma^*$. Note that $\pi$ is a continuous function from the convex and compact set ${\cal \CP}$ into itself.\footnote{\label{fn:existence}To establish continuity of $\pi$, write $\bm{\omega}_k(\tilde{\mathbf{\cp}})=\mathbf{u}_{k} + \beta \delta\mathbf{Q}_k |\mathbf{A}(\tilde{\mathbf{\cp}})|^{-1}\overline{\mathbf{A}}(\tilde{\mathbf{\cp}})\left[-\ln\tilde{\mathbf{\cp}}_K+\mathbf{u}_K\right]$, where $\mathbf{A}(\tilde{\mathbf{\cp}})\equiv \mathbf{I}-\delta\left[\beta\mathbf{Q}_K+(1-\beta)\sum_{k\in{\cal D}}\tilde{\mathbf{\CP}}_k{\mathbf{Q}_k}\right]$, its determinant $|\mathbf{A}(\tilde{\mathbf{\cp}})|$ is a finite sum of products of $J$ entries of $\mathbf{A}(\tilde{\mathbf{\cp}})$, and the elements of its adjugate $\overline{\mathbf{A}}(\tilde{\mathbf{\cp}})$ are the cofactors of $\mathbf{A}(\tilde{\mathbf{\cp}})$ and therefore finite sums  of products of $J-1$ of its entries. It follows that each element of $\bm{\omega}_k(\tilde{\mathbf{\cp}})$ is a finite sum of finite products of elements of $\tilde{\mathbf{\cp}}$ divided by another such finite sum, $|\mathbf{A}(\tilde{\mathbf{\cp}})|\neq 0$, and is therefore continuous in $\tilde{\mathbf{\cp}}$. Consequently, $(\bm{\omega}_1,\ldots,\bm{\omega}_K)$ is continuous. Because, obviously, $\varpi$ is continuous as well, $\pi=\varpi\circ(\bm{\omega}_1,\ldots,\bm{\omega}_K)$ is continuous. } So, by Brouwer's fixed point theorem, it has a fixed point $\mathbf{\cp}^*$, which equals $\mathbf{\cp}(\sigma^*)$ for some stationary perception perfect strategy $\sigma^*$. 
\end{proof}
Theorem \ref{th:existence} does not guarantee that the stationary perception perfect strategy $\sigma^*$ is unique, not even up to equivalence of the implied choice probabilities $\mathbf{\cp}(\sigma^*)$. Specifically, its proof relies on the application of Brouwer's fixed point theorem, which does not rule out that there are multiple fixed points $\mathbf{\cp}^*$.

\subsection{Identification}

Suppose that data allows us to determine conditional choice probabilities $\mathbf{\cp}$ and state transition probabilities $\mathbf{Q}_1,\ldots,\mathbf{Q}_K$, with 
\begin{align}
\label{eq:pStat}
\mathbf{\cp}&=\mathbf{\cp}(\sigma^*)
\end{align}
for some stationary perception perfect strategy $\sigma^*$.\footnote{As in the finite horizon case (see Footnote \ref{fn:dgp}), we typically observe the process $\{x_t.d_t\}$ over a finite number of periods. Under our model's assumptions, this process is a stationary Markov process. This can be tested. Unobserved heterogeneity, which may now include heterogeneous equilibrium selection, may render the data non-Markovian. Time variation in the parameters may lead to nonstationarity. Standard approaches to deal with heterogeneity can be applied to first identify agent-level choice and state transition probabilities. With these in hand, our results can be applied.} Under this assumption, Section \ref{sec:identification}'s (non-)identification results extend to the infinite horizon model.

First, we establish that, without further restrictions, the model is just identified if we fix the discount factors.
\begin{theorem}[{\bf Nonidentification}]
\label{th:nonidentStat}
For given $\mathbf{Q}_1,\ldots,\mathbf{Q}_K$; $\beta$; $\delta$; $\mathbf{u}_K$; and $\mathbf{\cp}$; there exists unique utilities $\mathbf{u}_1,\ldots,\mathbf{u}_{K-1}$  such that \eqref{eq:sppe}--\eqref{eq:pStat} hold.
\end{theorem}
\begin{proof}
Given $\delta$; $\beta$; $\mathbf{\cp}$; and $\mathbf{Q}_1,\ldots,\mathbf{Q}_K$; \eqref{eq:perceivedvalueStat} uniquely determines $\mathbf{v}(\sigma^*)$. With $\mathbf{u}_K$, $\beta$, $\delta$, and $\mathbf{Q}_K$; \eqref{eq:wStat} for $k=K$ subsequently determines $\mathbf{w}_K(\sigma^*)$. Then, given $\mathbf{\cp}$ and $\mathbf{w}_K(\sigma^*)$; $\mathbf{w}_{k}(\mathbf{\sigma^*}) =\mathbf{w}_{K}(\mathbf{\sigma^*})+\ln\mathbf{\cp}_{k}-\ln\mathbf{\cp}_{K}$, $k\in{\cal D}/\{K\}$, are the unique choice specific values consistent with  \eqref{eq:sppe}, \eqref{eq:ccpStat}, and \eqref{eq:pStat}. Given $\beta$; $\delta$; $\mathbf{Q}_k$, $k\in{\cal D}/\{K\}$; and $\mathbf{v}(\sigma^*)$; these can only be reconciled with \eqref{eq:wStat} by setting  $\mathbf{u}_{k}=\mathbf{w}_{k}(\sigma^*)-\beta \delta\mathbf{Q}_k\mathbf{v}(\sigma^*)$, $k\in{\cal D}/\{K\}$.
\end{proof}

Like Theorem \ref{th:nonident} before, Theorem \ref{th:nonidentStat} implies that further model restrictions are needed to identify $\beta$ and $\delta$. Moreover, it suggests that we concentrate the identification analysis on $\beta$ and $\delta$, as it ensures that we can always find utilities that rationalize the choice probabilities once we have identified those two parameters.

As in Section \ref{sec:identification}, we proceed by imposing the minimum of two exclusion restrictions on utility required for identification of $\beta$ and $\delta$: For some $k\in{\cal D}/\{K\}$,
\begin{align}
 u_{k}( x_{a,1}) &= u_{k}( x_{a,2}) ~~~\text{and}\label{eq:exclusionrestrictionsStatA}\\
 u_{k}( x_{b,1}) &=  u_{k}(x_{b,2}) \label{eq:exclusionrestrictionsStatB}
\end{align}
for  states $x_{a,1}, x_{a,2}, x_{b,1}, x_{b,2} \in \mathcal{X}$ such that $x_{a,1}\neq x_{a,2}$ and $x_{b,1}\neq x_{b,2}$.  We also again set $\mathbf{u}_K=\bm{0}$. A derivation as in Section \ref{sec:identification}, based on \eqref{eq:wStat} and \eqref{eq:perceivedvalueStat}, using \eqref{eq:sppe}, \eqref{eq:ccpStat}, and \eqref{eq:pStat} to substitute observed choice probabilties for unknown quantities, and \eqref{eq:exclusionrestrictionsStatA} and \eqref{eq:exclusionrestrictionsStatB} to difference out flow utility, gives
\begin{align}
\Delta^2_a\ln\mathbf{\cp}_k &= \beta \delta \Delta^2_a\mathbf{Q}_k\mathbf{A}(\beta,\delta)^{-1}\mathbf{m}~~~\text{and}~~~\Delta^2_b\ln\mathbf{\cp}_k = \beta \delta \Delta^2_b\mathbf{Q}_k\mathbf{A}(\beta,\delta)^{-1}\mathbf{m},\label{eq:systemofmomentsStatratio}
\end{align}
where\footnote{Note that $\mathbf{A}(\beta,\delta)$ is the same matrix as $\mathbf{A}(\tilde{\mathbf{\cp}})$ in Footnote \ref{fn:existence}. We allow this slight abuse of notation, instead of introducing a new symbol for this matrix, to clearly indicate both matrices are the same.} 
\begin{align*}
\Delta^2_a\ln\mathbf{\cp}_k&\equiv\ln\left(\frac{\cp_{k}(x_{a,1})}{\cp_{K}(x_{a,1})}\right) - \ln\left(\frac{\cp_{k}( x_{a,2})}{\cp_{K}(x_{a,2})}\right),\\
\Delta^2_a\mathbf{Q}_k&\equiv\mathbf{Q}_{k}(x_{a,1})-\mathbf{Q}_K(x_{a,1}) - \mathbf{Q}_{k}(x_{a,2})+\mathbf{Q}_K(x_{a,2}),
 \end{align*}
$\Delta^2_b\ln\mathbf{\cp}_k$ and $\Delta^2_b\mathbf{Q}_k$ are analogously defined, 
\begin{align*}
\mathbf{A}(\beta,\delta)&\equiv \mathbf{I}-\delta\left[\beta\mathbf{Q}_K+(1-\beta)\sum_{k\in{\cal D}}\mathbf{\CP}_k{\mathbf{Q}_k}\right], 
\end{align*}
and $\mathbf{m}\equiv-\ln\mathbf{\cp}_K$. As we noted before, the inverse of the $J\times J$ matrix $\mathbf{A}(\beta,\delta)$ exists, for all $\beta$ and $\delta$ in their domains. Using that $A(\beta,\delta)^{-1}=|\mathbf{A}(\beta,\delta)|^{-1}\overline{\mathbf{A}}(\beta,\delta)$, we can rewrite \eqref{eq:systemofmomentsStatratio} into
\begin{align}
|\mathbf{A}(\beta,\delta)|\Delta^2_a\ln\mathbf{\cp}_k &=\beta \delta \Delta^2_a\mathbf{Q}_k\overline{\mathbf{A}}(\beta,\delta)\mathbf{m}~~~\text{and}\label{eq:systemofmomentsStatA}\\
|\mathbf{A}(\beta,\delta)|\Delta^2_b\ln\mathbf{\cp}_k &=\beta \delta \Delta^2_b\mathbf{Q}_k\overline{\mathbf{A}}(\beta,\delta)\mathbf{m}.\label{eq:systemofmomentsStatB}
\end{align}
Here, the determinant $|\mathbf{A}(\beta,\delta)|$ is a polynomial of degree $J$ in both $\beta$ and $\delta$ (for a total degree of $2 J$) and the adjugate matrix $\overline{\mathbf{A}}(\beta,\delta)$ contains the cofactors of  $\mathbf{A}(\beta,\delta)$, which are polynomials of degree $J-1$ in both $\beta$ and $\delta$. Consequently, the left and right hand sides of \eqref{eq:systemofmomentsStatA} and \eqref{eq:systemofmomentsStatB} are polynomials of degree $J$ in both $\beta$ and $\delta$. Their coefficients depend on the observed choice and state transition probabilities only. So, under an assumption that excludes pathologies, we can mimic the proof of Theorem \ref{th:theorem1} to bound the cardinality of the identified set.
\begin{theorem}[{\bf Identified set}]\label{th:theoremStat}
Suppose that the exclusion restrictions in \eqref{eq:exclusionrestrictionsStatA} and \eqref{eq:exclusionrestrictionsStatB} hold and that the polynomial differences between the left and right hand sides of \eqref{eq:systemofmomentsA} and \eqref{eq:systemofmomentsB} are nonzero and have no common factors. Then the identified set contains at most $J^2$ elements. 
\end{theorem}

Finally, note that the above analysis also implies identification results for the special case with geometric discounting ($\beta=1$). For given $\beta\in[0,1)$, the difference between the left and right hand sides of \eqref{eq:systemofmomentsStatA} is a $J$'th degree polynomial in $\delta$, with coefficients that are known from the data. Provided that this polynomial is nonzero, by the fundamental theorem of algebra, at most $J$ discount factors $\delta$ satisfy \eqref{eq:systemofmomentsStatA} for given $\beta$. So, under exclusion restriction \eqref{eq:exclusionrestrictionsStatA} and assumptions that ensure that the implied polynomial is nonzero, the identified set of the special case with geometric discounting contains at most $J$ elements. This sharpens \cites{qe20:abbringdaljord} Theorem 1 and Corollary 1, which only established that the identified set is discrete, and finite if discount factors near 1 are excluded. For future reference, we explicitly state this result as\footnote{Note that \citet{qe20:abbringdaljord} denote the geometric discount factor with $\beta$ instead. Also, \citeauthor{qe20:abbringdaljord} allowed more generally for an exclusion restriction across either states, or choices, or both. Our analysis here directly extends to such a more general exclusion restriction. We keep this implicit for notational simplicity.}
\begin{theorem}[{\bf Geometric special case}]\label{th:geometric}
Suppose that $\beta=1$, exclusion restriction \eqref{eq:exclusionrestrictionsStatA} holds, and either $\Delta^2_a\ln\mathbf{\cp}_k\neq 0$ or $\Delta^2_a\mathbf{Q}_k\mathbf{m}\neq 0$. Then the identified set contains at  most $J$ elements.
\end{theorem}
Both Theorems \ref{th:theoremStat} and \ref{th:geometric} rely on the fact that ${\cal X}$ is finite. In contrast, we conjecture that, under some regularity conditions, \cites{qe20:abbringdaljord} Theorem 1 and Corollary 1 extend to continuous state variables. See \cite{blevins2014} for a relevant approach to this case.

\section{Empirical application}
\label{s:margarine}
Brand loyalty has been studied in static choice models since the 1980's. \cite{daljorddubekong2020} notes notes that in these models, a positive coefficient on repeated brand purchases has been interpreted as evidence of brand loyalty. If we leave aside the confounding of persistent heterogeneity and state dependence in these models and accept this interpretation of persistence in brand choice, it still does not explain why consumers are brand loyal. 

\cite{farrellklemperer07} considers brand loyalty as a classic case of a switching cost that creates a lock-in. Switching costs may be pecuniary, such as terminating a mobile subscription plan, or psychological, such as breaking a habitual brand choice. Using a \cite{beckermurphy88} rational addiction like argument, \cite{gordonsun15} argues that if consumers are aware that their brand choice creates a lock-in, then forward-looking consumers would take the lock-in into account when choosing among brands. In a consumer packaged goods context, a forward-looking consumer may be wary of choosing a brand with a high average price and little price variation. \cite{daljorddubekong2020} tested for forward looking behaviour in a dynamic choice model with brand loyalty and geometric discounting. It found a value of $\delta = 0.41$ on a weekly level, statistically significant, but economically small. In this application, we extend \cite{daljorddubekong2020}'s analysis to allow for present-biased preferences using the identification argument of the previous section applied to the same data.  

The data are from the margarine household purchase panel of Nielsen-Kilts HMS that was used in \cite{dubeetal2010}. 
[Description and summary statistics of proprietary data suppressed.]

In each period, the household makes a discrete choice $d_{i,t} \in \mathcal{D}$. The choice set has four brands, indexed 1 to 4, and the reference choice $K = 5$ is purchasing some spreadable good other than margarine. The state variables are  prices $\bm{p} = [p_1, \hdots, p_4] \in \mathbb{R}^4_+$ and brand loyalty $l \in \mathcal{D}\backslash \{K\}$. The prices are assumed to follow a first-order Markov process that is exogenous to the household's choices. The brand loyalty state $l$ is controlled by the choice and is equal to the last purchase of the inside goods.
\begin{align}
l_{t+1} = 
\begin{cases}
d_{t}  & \text{ if } d_{t} \neq K \\
l_{t} & \text{ otherwise}.
\end{cases}
\end{align}
The utilities are
\begin{align}\label{eq:margarineutility}
u_{k}(x_t) & = 
\begin{cases}
u_{k}(p_{k,t}, l_{t}) & \text{ if } k \neq K \\
0 & \text{ otherwise.} \\
\end{cases}
\end{align}
This utility function allows for brand loyalty for some products (positive state dependence) and variety seeking in other products (negative state dependence).  This specification is more flexible than the parametric counterpart in \cite{dubeetal2010} which uses $\mathbbm{1}(k = l)$ as the loyalty state variable and a parametric form that restricts state dependence to be either positive or negative for all products. Substituting in the utilities in \eqref{eq:margarineutility}, the current choice specific value function has the structure of \eqref{eq:wStatscal} and \eqref{eq:perceivedvalueStatscal}, with choice probabilities given by \eqref{eq:ccpStat}.

\cite{daljorddubekong2020} observes that in the standard dynamic differentiated products model, the utility of purchasing one of the inside goods excludes the prices of rival products, such that
 \begin{align}
 u_k(p_k, \bm{p}_{-k}, l) = u_k(p_k,  l).
\end{align}
A pair of states $x_{1} = [p_k, \bm{p}_{-k}, l]$  and $x_{2}=[p_k, \bm{p}'_{-k}, l]$  for which the own price and brand loyalty state are constant, but where at least one rivals's product price shifts ($\bm{p}_{-k} \neq \bm{p}'_{-k}$) satisfies the exclusion restrictions in \eqref{eq:exclusionrestrictionsStatA}  or \eqref{eq:exclusionrestrictionsStatB} for choice $k$. The discount factor is therefore identified in these models if there are at least two such pairs of states that satisfy the regularity conditions. A theoretical virtue of this identification strategy is that it is always available in the discrete choice models used to analyze CPG demand without further assumptions. 

The regularity conditions in Theorem \ref{th:theoremStat} require that the moment condition corresponding to a given exclusion restriction is non-zero. A necessary condition for a non-zero right hand side of either \eqref{eq:systemofmomentsStatA} or \eqref{eq:systemofmomentsStatB} is
\begin{align}\label{eq:nonzerocondition}
\Delta^2_a\mathbf{Q}_k&\equiv\mathbf{Q}_{k}(x_{1})-\mathbf{Q}_K(x_{1}) - \mathbf{Q}_{k}(x_{2})+\mathbf{Q}_K(x_{2}) \neq 0.
\end{align}
Since prices are assumed to evolve independently of household choices, then unless choice $k$ shifts the next period's loyalty state $l$, it follows that
\begin{align}
\mathbf{Q}_{k}(x_{1})-\mathbf{Q}_K(x_{1})  &= \mathbf{Q}_{k}(\bm{p}_{1}, l_{1}) -\mathbf{Q}_K(\bm{p}_{1}, l_{1}) = 0\\
\mathbf{Q}_{k}(x_{2})-\mathbf{Q}_K(x_{2})  &= \mathbf{Q}_{k}(\bm{p}_{2}, l_{2}) -\mathbf{Q}_K(\bm{p}_{2}, l_{2}) = 0.
\end{align}
So if $k = l$, the non-zero condition in  \eqref{eq:nonzerocondition} fails and the corresponding moment condition is not informative, though the exclusion restriction holds. We  define the set of pairs of states that satisfy necessary conditions for identification
\begin{align}
\mathcal{X}^{id}=\left\{ \left\{ x_1,x_2\right\} \in\mathcal{X}^{2}:p_{k,1}=p_{k,2}\,, \,\bm{p}_{-k,1}\neq\bm{p'}_{-k, 2},l_1 = l_2 \equiv l, l\neq k \right\} \,\text{ for }~\,k\in\mathcal{D}\backslash\left\{ 0\right\} .\label{eq:exclusionset}
\end{align}
As in \cite{daljorddubekong2020}, the prices of the four brands are discretized in 30 bins. Along with four brand loyalty states, this gives a total of 120 states.  We estimate the parameters $\beta$ and $\delta$ from the moment conditions that can be generated from $\mathcal{X}^{id}$ by minimum distance with quadratic loss. The parameters $\beta$ and $\delta$ are identified if  $\mathcal{X}^{id}$ contains at least two unique pairs. In this application, we have 3400. 

The geometric discount factor was estimated without restrictions its domain and replicates the result in \cite{daljorddubekong2020}. 
[Estimation results of the hyperbolic model on proprietary data suppressed.]

\section{Concluding remarks}

The previous section showed that the exclusion restriction approach that identifies geometric time preferences in dynamic discrete choice models formally extends to present-biased time preferences. The identification argument for the geometric case is intuitive. The moment condition derived from the exclusion restriction captures the choice response to variation in future values for a fixed current utility. The intuition behind the identification of present-bias is less obvious. 

The defining feature of present-bias that the identification strategy must capture is preference reversals. A well-known example of preference reversals is Thaler's apples: while most people prefer an apple today to two apples tomorrow, the same people also prefer two apples one year and one day from now to one apple one year from now. Such observed choice contrasts are direct and intuitive evidence of preference reversals that are commonly used in lab studies of time preferences, see e.g. \cite{marziliericsonetal15}. 

Inferring preference reversals from observed choices in dynamic discrete choice models requires utility restrictions. For instance, Thaler's example can be rationalized with non-stationary utilities without invoking preference reversals. Identification of hyperbolic discount functions using exclusion restrictions on utilities relies on evidence of preference reversals similar to Thaler's example, but in a less transparent way. Each moment condition represents the difference in values of a particular choice contrasts that involves a sequence of expected choices. If we fix $\beta = 1$, then $\delta$ is determined from one of the moment conditions. Once $\delta$ is determined,  $\mathbf{u}$ is uniquely determined. For a particular set of $\delta$ and $\mathbf{u}$, we can test if these preferences are consistent with the second moment condition, which then has no free parameters. If the second moment condition is not satisfied at $\beta = 1$, it is because we have a preference reversal: a $\delta$ and a $\mathbf{u}$ that rationalizes the choice contrast in the first moment conditions is inconsistent with the revealed preference for the choice contrast in the second moment condition. By allowing $\beta \in (0,1)$ and $\delta$ to simultaneously determine both moment conditions, the observed preference reversal informs the present bias parameter.


\clearpage
\appendix
\counterwithin{figure}{section}
\pdfbookmark[0]{References}{pdfbm:refs}
\bibliographystyle{chicago}
\bibliography{alljaap2}

\section{Identification and inference in a three period model}
In this example, we assume binary choice. We set $t_{a,1} = t_{a,2} = t_a$ and $t_{b,1} = t_{b,2} = t_b$ and assume the exclusion restrictions
\begin{align}
u_{1,t_a}(x_{a,1}) = &u_{1,t_a}(x_{a,2}) \\
u_{1,t_b}(x_{b,1}) = &u_{1,t_b}(x_{b,2})
\end{align}
The two exclusion restrictions lead to the two moment conditions 
\begin{align}
		&\ln\left(\frac{\cp_{1,t_a}(x_{a,1})}{\cp_{2,t_a}(x_{a,1})}\right) - \ln\left(\frac{\cp_{1,t_a}(x_{a,2})}{\cp_{2,t_a}(x_{a,2})}\right) = \nonumber \\
		 & \beta \delta \left[\mathbf{Q}_{1}(x_{a,1}) - \mathbf{Q}_K(x_{a,1})
		 -\mathbf{Q}_{1}(x_{a,2}) + \mathbf{Q}_K(x_{a,2})\right]
		 \left[\mathbf{m}_{t_a+1} + \delta \mathbf{Q}^{pb}_{t+1}\mathbf{v}_{t_a+2}  \right], \\
&\ln\left(\frac{\cp_{1,t_b}(x_{b,1})}{\cp_{2,t_b}(x_{b,1})}\right) - \ln\left(\frac{\cp_{1,t_b}(x_{b,2})}{\cp_{2,t_b}(x_{b,2})} \right)  =  \nonumber\\
		& \beta \delta \left[\mathbf{Q}_{1}(x_{b,1})  - \mathbf{Q}_K(x_{b,1})
		 -\mathbf{Q}_{1}(x_{b,2})+ \mathbf{Q}_k(x_{b,2})\right]\left[\mathbf{m}_{t_b+1} + \delta \mathbf{Q}^{pb}_{t_b+1}\mathbf{v}_{t_b+2} \right] .
\end{align}

\subsection*{Period $T-1$} 
We first show that present-biased discount functions can not be identified from only two periods of observed choices and states. 
Let $t_a = t_b = T-1$, then define  
\begin{align*}
\Delta \mathbf{Q}_1(x_a) & = \left[\mathbf{Q}_1(x_{a,1}) - \mathbf{Q}_K(x_{a,1})- \mathbf{Q}_1(x_{a,2}) + \mathbf{Q}_K(x_{a,2})\right],
\end{align*}
and analogously for $\Delta \mathbf{Q}_1(x_b)$. Next, define 
\begin{align*}
\Delta \ln(\cp_{1, T-1}(x_a)) & = \ln\left(\frac{\cp_{1,T-1}(x_{a,1})}{\cp_{2,T-1}(x_{a,1})}\right) - \ln\left(\frac{\cp_{1,T-1}( x_{a,2})}{\cp_{2,T-1}(x_{a,2})}\right).
\end{align*}
The moment conditions in \eqref{eq:systemofmomentsA} can now be written
\begin{align}
		\Delta \ln(\cp_{1, T-1}(x_a)) & =  \beta \delta \Delta \mathbf{Q}_1(x_a) \mathbf{m}_{T}  \\
		\Delta \ln(\cp_{1, T-1}(x_b)) &=  \beta \delta \Delta \mathbf{Q}_1(x_b) \mathbf{m}_{T} 
\end{align}
The two polynomials are clearly linearly dependent. This also an example of a common factor. Since the parameters $\beta$ and $\delta$ are interchangeable in both moment conditions, they can not be separately identified with only two periods of data. Their product is however point identified. 

\subsection*{Period $T-2$}
With three periods of data, the discount function parameters are formally set identified. 
Let $t_a = t_b = T-2$. The moment conditions are
\begin{align*}
		\Delta \ln(\cp_{1, T-2}(x_a)) & =   \beta \delta  \Delta \mathbf{Q}_1(x_a) \left[\mathbf{m}_{T-1} + \delta\mathbf{Q}^{pb}_{T-1}\mathbf{m}_{T} \right]   \\
				\Delta \ln(\cp_{1, T-2}(x_b)) & =   \beta \delta  \Delta \mathbf{Q}_1(x_b) \left[\mathbf{m}_{T-1} + \delta\mathbf{Q}^{pb}_{T-1}\mathbf{m}_{T} \right]  
\end{align*}
Writing out the terms, we get
\begin{align}
		 \Delta \ln(\cp_{1, T-2}(x_a)) = & \beta \delta \Delta \mathbf{Q}_1(x_a)\mathbf{m}_{T-1}  + \beta \delta^2 \Delta \mathbf{Q}_1(x_a)\overline{\mathbf{Q}}_{T-1}\mathbf{m}_{T} \nonumber + \\
		 &  \beta^2 \delta^2 \Delta \mathbf{Q}_1(x_a)\left[ \overline{\mathbf{Q}}_{T-1}- \mathbf{Q}_2 \right]\mathbf{m}_T    \\
		 \Delta \ln(\cp_{1, T-2}(x_b)) = & \beta \delta \Delta \mathbf{Q}_1(x_b)\mathbf{m}_{T-1}  + \beta \delta^2 \Delta \mathbf{Q}_1(x_b)\overline{\mathbf{Q}}_{T-1}\mathbf{m}_{T} + \nonumber \\
		 & \beta^2 \delta^2 \Delta \mathbf{Q}_1(x_b)\left[ \overline{\mathbf{Q}}_{T-1}- \mathbf{Q}_2 \right]\mathbf{m}_T 
\end{align}
We first note that the only term for which $\beta$ and $\delta$ are not interchangeable in period $T-2$ is $\beta \delta^2 \Delta \mathbf{Q}_1(x_a)\overline{\mathbf{Q}}_{T-1}\mathbf{m}_T$. The set identification of $\beta$ and $\delta$ therefore relies on a higher order interaction term. These terms are furthermore likely to be highly correlated in finite samples which suggests that precise estimation of the two parameters separately may be hard to achieve. We illustrate this point with a simulation below.

\subsection{Estimation routine}
We estimate $\beta$ and $\delta$ from  the sample counterparts to the moment conditions in \eqref{eq:systemofmomentsA} and \eqref{eq:systemofmomentsB} by minimum distance. Holding the choice fixed at some $k \in \mathcal{D}\backslash \{K\}$,  a pair of periods $t$ and $t'$ and a pair of states $x_{1}$ and $x_{2}$ give the exclusion restriction $u_t(x_1) = u_{t'}(x_2)$. The corresponding moment is 
\begin{align}\label{eq:criterionmoment}
\mathbf{\psi}(\beta, \delta; x, t) = & \frac{\cp_{k,t}(x_{1})}{\cp_{K,t}(x_{1})} - \frac{\cp_{k,t'}(x_{2})}{\cp_{K,t'}(x_{2})}- \nonumber\\
&\beta\delta \left([\mathbf{Q}_k(x_{1}) - \mathbf{Q}_K(x_{2})] \mathbf{v}_{t+1} - [\mathbf{Q}_k(x_{2}) - \mathbf{Q}_K(x_{2})] \mathbf{v}_{t'+1} \right) 
\end{align}
where $\mathbf{v}_{t} = \mathbf{m}_t + \delta \mathbf{Q}^{pb}_t \mathbf{v}_{t+1}$.  We denote the vector of moments which has one element for each exclusion restriction  $\mathbf{\psi}(\beta, \delta;.,.)$. The minimum distance criterion is
\begin{align*}
S(\beta, \delta) & = \mathbf{\psi} \mathbf{W} \mathbf{\psi}'
\end{align*}
for a weight matrix $\mathbf{W}$. The gradient and the Hessian of the criterion function are given in the appendix.

\def\N{ 1000000}
\def\J{     6}
\def\T{     3}
\def\II{   100}
\def\betatrue{0.80}
\def\deltatrue{0.50}
\def\utrue{\left[\begin{array}{ccc}  1.00 &  -1.00 &   1.00 \\   1.00 &   2.00 &   1.00\\   1.00 &   2.00 &   4.00\\   1.00 &  -1.00 &   4.00\\  4.00 &   2.00 &   1.00\\  1.00 &   5.00 &   3.00\end{array}\right]}
\def\uoneexcluded{  1.00}
\def\Qone{\left[\begin{array}{cccccc}  0.19 &   0.22 &   0.06 &   0.28 &   0.06 &   0.19 \\  0.11 &   0.32 &   0.07 &   0.11 &   0.14 &   0.25 \\  0.28 &   0.11 &   0.17 &   0.28 &   0.06 &   0.11\\  0.21 &   0.14 &   0.24 &   0.24 &   0.07 &   0.10\\  0.03 &   0.24 &   0.24 &   0.24 &   0.22 &   0.03\\  0.10 &   0.14 &   0.10 &   0.19 &   0.05 &   0.43 \end{array}\right]}
\def\Qtwo{\left[\begin{array}{cccccc}  0.25 &   0.19 &   0.12 &   0.12 &   0.12 &   0.19 \\  0.08 &   0.08 &   0.31 &   0.15 &   0.23 &   0.15 \\  0.27 &   0.07 &   0.27 &   0.07 &   0.20 &   0.13\\  0.23 &   0.23 &   0.31 &   0.08 &   0.08 &   0.08\\  0.19 &   0.25 &   0.12 &   0.06 &   0.25 &   0.12\\  0.19 &   0.12 &   0.19 &   0.19 &   0.25 &   0.06 \end{array}\right]}
\def\betahattruedata{0.80}
\def\deltahattruedata{0.50}
\def\betahatsamplingvariation{0.86}
\def\deltahatsamplingvariation{0.44}
\def\deltahatgeometricsamplingvariation{0.40}

\subsection{Simulation}
We set the number of states $J = \J$, for $T = \T$, and draw data for $N = \N$ agents. The discount parameters are set $\beta = \betatrue$ and $\delta = \deltatrue$. The exclusion restrictions $u_{1,1}(x_1) = u_{1,1}(x_2) = \uoneexcluded$ are imposed in estimation. The utilities are 
\begin{align*}
\mathbf{u}_1 & = \utrue
\end{align*}
and the transitions are drawn randomly from the true transitions
\begin{align*}
\mathbf{Q}_1= \Qone \\
\mathbf{Q}_2 = \Qtwo
\end{align*}
We first confirm that $\beta$ and $\delta$ are identified. We use the true choice probabilities and true transition distributions to recover $\beta$ and $\delta$ up to numerical precision at $\hat{\beta} = \betahattruedata$ and $\hat{\delta} = \deltahattruedata$.
\\
\\
Figure \ref{fig:criterionbetadelta} plots the criterion for $\beta$ and $\delta$, holding $\delta$ and $\beta$, respectively, at their true values, using the true choice data. The plot shows no clear basin around the minimum, but instead a banana shaped trough. The trough points to issues of inference in finite samples. A similar observation was made in \cite{laibsonetal2007} for  a lifecycle consumption model with $\beta \delta$ preferences and continuous choices, see its Figure 1. 
\begin{figure}[h] 
   \centering
   \includegraphics[width=10cm]{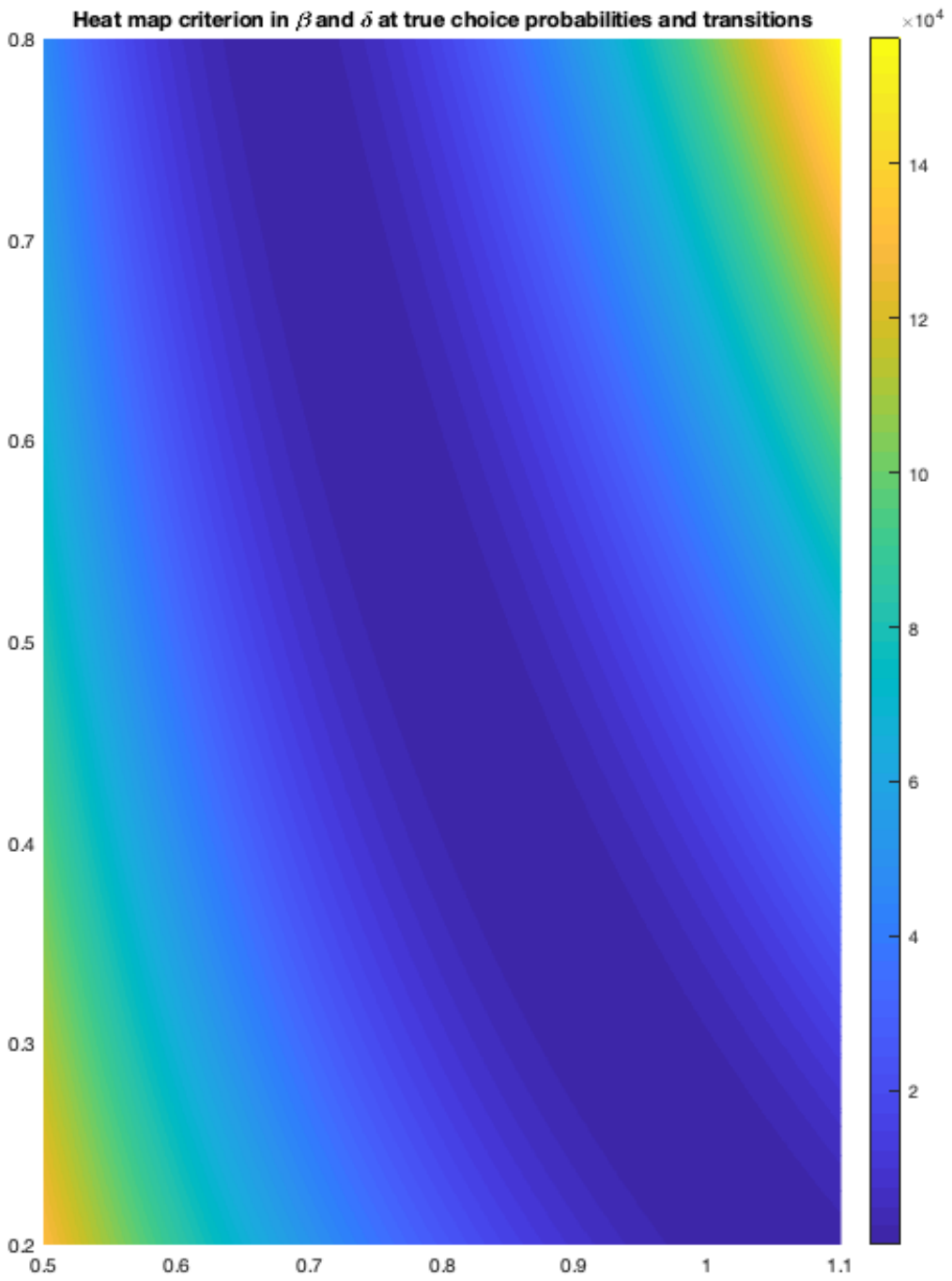} 
   \caption{\footnotesize  Heat map of the criterion function for the hyperbolic model using true choice data (no sampling variation). }
      \label{fig:criterionbetadelta}
\end{figure}
\\
\\
We next use choice data with sampling variation. In Figure \ref{fig:jointdistributionestimates}, we plot $\beta$ and $\delta$ estimates from $\II$ data sets drawn from the same DGP. The estimates are seen to lie along a hyperbole that is implied by the product of their true values $\beta  = \frac{\betatrue* \deltatrue}{\delta}$, similar to the trough in the heat map in Figure \ref{fig:criterionbetadelta}. The scatterplot shows that though the parameters are imprecisely estimated separately (the swarm of points stretch along the hyperbole), the products of the parameters are relatively more precisely recovered (the variation around the hyperbole). This points to a practical difficulty in recovering hyperbolic discount function parameters precisely in observational data using our exclusion restrictions. 
\begin{figure}[h] 
   \centering
   \includegraphics[width=10cm]{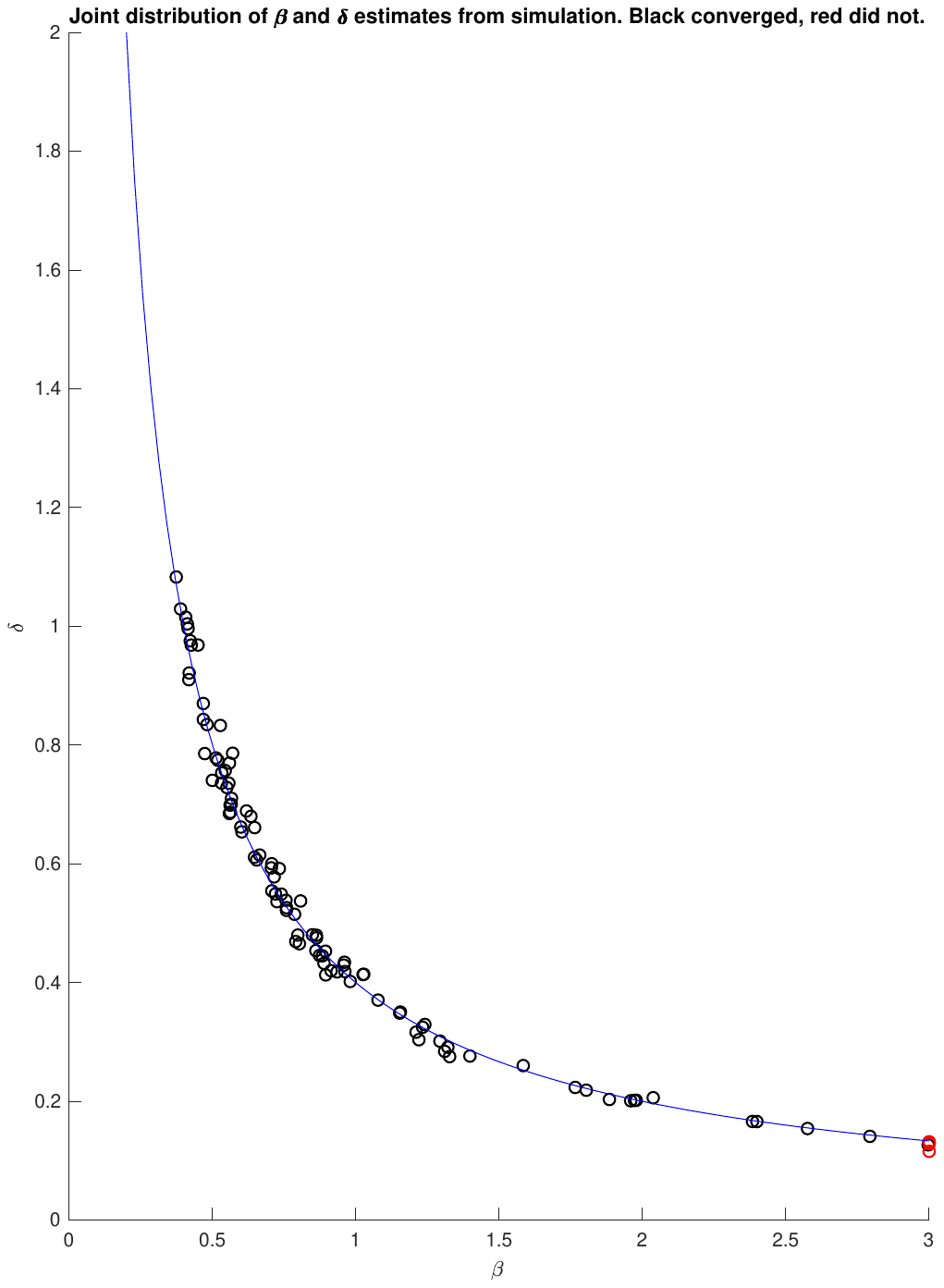} 
   \caption{\footnotesize Estimates of $\beta$ and $\delta$ from data with sampling variation. }
      \label{fig:jointdistributionestimates}
\end{figure}
Finally, we estimate an exponential discount function using data generated by a DGP with $\beta = \betatrue$ and $\delta = \deltatrue$. We expect the estimate of $\delta$ to be close to $\betatrue \times \deltatrue$ and precisely estimated. The estimate is $\deltahatgeometricsamplingvariation$. The criterion is given in Figure \ref{fig:criteriondelta}.
\begin{figure}[h] 
   \centering
   \includegraphics[width=10cm]{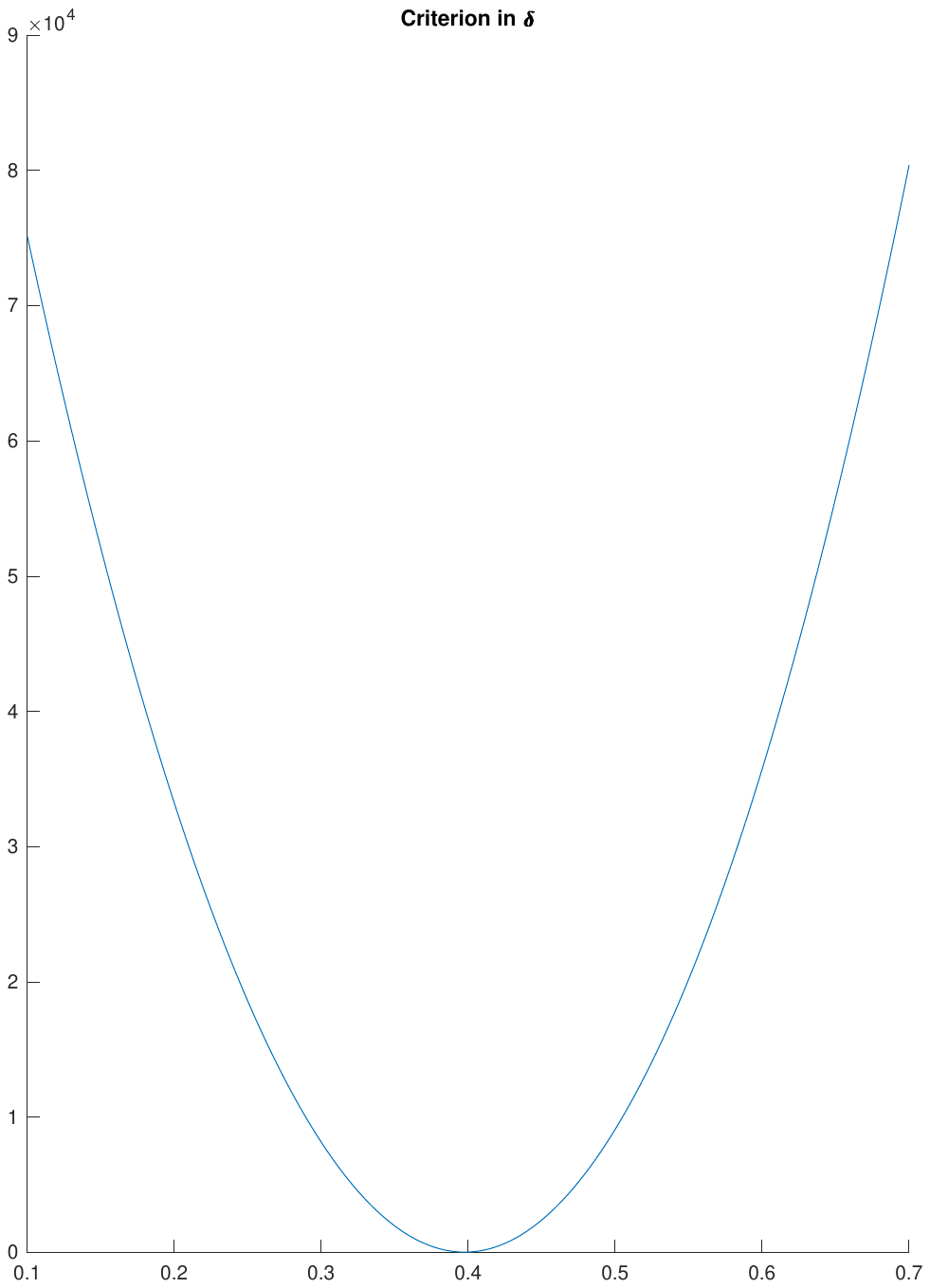} 
   \caption{\footnotesize  Plot of the criterion function for the geometric model using choice data with sampling variation.}
      \label{fig:criteriondelta}
\end{figure}

\end{document}